\documentclass[11pt,a4paper]{article}
\usepackage{vmargin, cite}
\setmarginsrb{1.2in}{1.2in}{1.2in}{1.2in}{0mm}{0mm}{5mm}{5mm}

\usepackage{mystyle}

\usepackage[utf8x]{inputenc}
\usepackage{amsfonts,amsmath,amssymb}
\usepackage{graphicx}
\usepackage{amsmath}
\usepackage{color}
\usepackage{mathtools}
\usepackage{xspace}

\DeclareMathOperator{\edge}{edge}
\DeclareMathOperator{\maxNE}{maxNE}
\DeclareMathOperator{\minNE}{minNE}

\newcommand{\NE}{NE\xspace}
\newcommand{\GE}{GE\xspace}
\newcommand{\kGE}{\textit{k-}GE\xspace}
\newcommand{\kNE}{\textit{k-}NE\xspace}

\title{Network Creation Games: Think Global -- Act Local\thanks{Accepted at the 40th International Symposium on Mathematical Foundations of Computer Science (MFCS) 2015.}\\{\small(full version)}}
\author{Andreas Cord-Landwehr\thanks{This work was partially supported by the German Research Foundation (DFG) within the Collaborative Research Centre ``On-The-Fly Computing'' (SFB 901).}\\Heinz Nixdorf Institute \& Department of Computer Science\\University of Paderborn\\ \email{andreas.cord-landwehr@uni-paderborn.de}
\and 
Pascal Lenzner\\
Department of Computer Science\\Friedrich-Schiller-University Jena\\ \email{pascal.lenzner@uni-jena.de}}

\date{}

\begin{document}

\maketitle

\begin{abstract}
\noindent We investigate a non-cooperative game-theoretic model for the formation of communication networks by selfish agents.
Each agent aims for a central position at minimum cost for creating edges.
In particular, the general model (Fabrikant et al., PODC'03) became popular for studying the structure of the Internet or social networks.
Despite its significance, locality in this game was first studied only recently (Bil\`o et al., SPAA'14), where a worst case locality model was presented, which came with a high efficiency loss in terms of quality of equilibria.
Our main contribution is a new and more optimistic view on locality:
agents are limited in their knowledge and actions to their local view ranges, but can probe different strategies and finally choose the best.
We study the influence of our locality notion on the hardness of computing best responses, convergence to equilibria, and quality of equilibria. Moreover, we compare the strength of local versus non-local strategy-changes.
Our results address the gap between the original model and the worst case locality variant.
On the bright side, our efficiency results are in line with observations from the original model, yet we have a non-constant lower bound on the price of anarchy.
\end{abstract}

\section{Introduction}
Many of today's networks are formed by selfish and local decisions of their participants.
Most prominently, this is true for the Internet, which emerged from the uncoordinated peering decisions of thousands of autonomous subnetworks.
Yet, this process can also be observed in social networks, where participants selfishly decide with whom they want to interact and exchange information.
\emph{Network Creation Games} (NCGs) are known as a widely adopted model to study the evolution and outcome of such networks.
In the last two decades, several such game variants were introduced and analyzed in the fields of economics, e.g.\ Jackson \& Wolinsky~\cite{jackson1996strategic}, Bala \& Goyal~\cite{BG00}, and theoretical computer science, e.g.\ Fabrikant et al.~\cite{Fab03}, Corbo \& Parkes~\cite{CP05}.

In all of these models, the acting agents are assumed to have a global knowledge about the network structure on which their decisions are based.
Yet, due to the size and dynamics of those networks, this assumption is hard to justify.
Only recently, Bilò et al.\ \cite{Bil14local} introduced the first variant of the popular (and mathematically beautiful) model by Fabrikant et al.\ \cite{Fab03} that explicitly incorporates a locality constraint.
In this model, the selfish agents are nodes in a network which can buy arbitrary incident edges.
Every agent strives to maximize her service quality in the resulting network at low personal cost for creating edges.
The locality notion by Bil\`o et al.~\cite{Bil14local} incorporates a worst case view on the network, which limits agents to know only their neighborhood within a bounded distance.
The network structure outside of this view range is assumed to be worst possible. In particular, the assumed resulting cost for any agent's strategy-change are estimated as the worst case over all possible network structures outside of this view.
In a follow-up work \cite{Bil14traceroute}, this locality notion was extended by enabling agents to obtain information about the network structure by traceroute based strategies.
Interestingly, in both versions the agents' service quality still depends on the whole network, which is actually a realistic assumption.
As their main result, Bil\`o et al.\ show a major gap in terms of efficiency loss caused by selfish behavior compared to the original non-local model.

In this paper, we extend the investigation of the influence of locality in NCGs by studying a more optimistic but still very natural model of locality. Our model allows us to map the boundary of what locally constrained agents can actually hope for.
Thereby, we close the gap between the non-local original version and the worst case locality models by Bil\`o et al. Besides studying the impact of locality on the outcomes' efficiency, we also analyze the impact on the computation of best response strategies and on the dynamic properties of the induced network creation process and we compare the strength of local versus non-local strategy-changes from an agent's perspective.

\paragraph{Our Locality Approach.}
We assume that an agent $u$ in a network only has complete knowledge of her $k$-neighborhood.
Yet, in contrast to Bil\`o et al.~\cite{Bil14local,Bil14traceroute}, besides knowing the induced subnetwork of all the agents that have distance of at most $k$ to $u$, the agent can, e.g. by sending messages, judge her actual service quality that would result from a strategy-change.
That is, we assume rational agents that ``probe'' different strategies, get direct feedback on their cost and finally select the best strategy.
It is easy to see that allowing the agents to probe all available strategy-changes within their respective $k$-neighborhood and then selecting the best of them is equivalent to providing the agents with a global view, but restricting their actions to $k$-local moves.
Here, a \emph{$k$-local move} is any combination of (1) removing an own edge, (2) swapping an own edge towards an agent in the $k$-neighborhood, and (3) buying an edge towards an agent in the $k$-neighborhood.  

Depending on the size of the neighborhood, probing all $k$-local moves may be unrealistic since there can be exponentially many such strategy-changes. To address this issue, we will also consider \emph{$k$-local greedy moves}, which are $k$-local moves consisting only of exactly one of the options (1)--(3). It is easy to see that the number of such moves is quadratic in the number of vertices in the $k$-neighborhood and especially for small $k$ and sparse networks this results in a small number of probes. 

We essentially investigate the trade-off between the cost for eliciting a good local strategy by repeated probing and the obtained network quality for the agents. For this, we consider the extreme cases where agents either probe all $k$-local moves or only a polynomial fraction of them. Note that the former sheds light on the networks created by the strongest possible locally constrained agents.   

\paragraph{Model and Notation.}
We consider the NCGs as introduced by Fabrikant et al.~\cite{Fab03}, where $n$ agents $V$ want to create a connected network among themselves.
Each agent selfishly strives for minimizing her cost for creating network links, while maximizing her own service quality in the network.
All edges in the network are undirected, have unit length and agents can create any incident edge for the price of $\alpha>0$, where $\alpha$ is a fixed parameter of the game.
(Note that in our illustrations, we depict the edge-ownerships by directing the edges away from their owners, yet still understand them as undirected.)
The strategy $S_u \subseteq V\setminus\{u\}$ of an agent $u$ determines which edges are bought by this agent, that is, agent $u$ is willing to create (and pay for) all the edges $ux$, for all $x \in S_u$. 
Let $\mathcal{S}$ be the $n$-dimensional vector of the strategies of all agents, then $\mathcal{S}$ determines an undirected network $G_{\mathcal{S}} = (V,E_{\mathcal{S}})$, where for each edge $uv \in E_{\mathcal{S}}$ we have $v \in S_{u}$ or $u \in S_{v}$.
If $v \in S_{u}$, then we say that agent $u$ is the \emph{owner} of edge $uv$, otherwise, if $u \in S_{v}$, then agent $v$ owns the edge $uv$. We assume throughout the paper that each edge in $E_{\mathcal{S}}$ has a unique owner, which is no restriction, since no edge can have two owners in any equilibrium network. In particular, we assume that the cost of an edge cannot be shared and every edge is fully paid by its owner. 
With this, it follows that there is a bijection between strategy-vectors $\mathcal{S}$ and networks $G$ with edge ownership information.
Thus, we will use networks and strategy-vectors interchangeably, i.e., we will say that a network $(G,\alpha)$ is in equilibrium meaning that the corresponding strategy-vector $\mathcal{S}$ with $G = G_{\mathcal{S}}$ is in equilibrium for edge-price $\alpha$.
The edge-price $\alpha$ will heavily influence the equilibria of the game, which is why we emphasize this by using $(G,\alpha)$ to denote a network $G$ with edge ownership information and edge-price $\alpha$.
Let $N_k(u)$ in a network $G$ denote the set of all nodes in $G$ with distance of at most $k$ to $u$.
The subgraph of $G$ that is induced by $N_k(u)$ is called the \emph{$k$-neighborhood} of $u$.

There are two versions for the cost function of an agent, which will yield two different games called the \textsc{Sum}-NCG~\cite{Fab03} and the \textsc{Max}-NCG~\cite{De07}.
The cost of an agent $u$ in the network $G_{\mathcal{S}}$ with edge-price~$\alpha$ is $\cost_u(G_{\mathcal{S}},\alpha) = \edge_u(G_{\mathcal{S}},\alpha) + \dist_u(G_{\mathcal{S}}) = \alpha|S_u|+\dist_u(G_{\mathcal{S}})$, where in the \textsc{Sum}-NCG we have that $\dist_u(G_{\mathcal{S}}) = \sum_{w\in V} d_{G_{\mathcal{S}}}(u,w)$, if $G_{\mathcal{S}}$ is connected and $\dist_u(G_{\mathcal{S}}) = \infty$, otherwise.
Here, $d_{G_{\mathcal{S}}}(u,w)$ denotes the length of the shortest path between $u$ and $w$ in $G_{\mathcal{S}}$.
Since all edges have unit length, $d_{G_{\mathcal{S}}}(u,w)$ is the hop-distance between $u$ and $w$.
In the \textsc{Max}-NCG the sum-operator in $\dist_u(G_{\mathcal{S}})$ is replaced by a max-operator.\footnote{Throughout this paper, we will only consider connected networks as they are the only ones which induce finite costs. }

A network $(G_{\mathcal{S}},\alpha)$ is in \emph{Pure Nash Equilibrium} (NE), if no agent can unilaterally change her strategy to strictly decrease her cost.
Since the \NE has undesirable computational properties, researchers have considered weaker solution concepts for NCGs.
$G_{\mathcal{S}}$ is in \emph{Asymmetric Swap Equilibrium} (ASE)~\cite{MS12}, if no agent can strictly decrease her cost by swapping one own edge.
Here, a swap of agent $u$ is the replacement of one incident edge $uv$ by any other new incident edge $uw$, abbreviated by $uv\to uw$.
Note that this solution concept is independent of the parameter $\alpha$ since the number of edges per agent cannot change.
A network $(G_{\mathcal{S}},\alpha)$ is in \emph{Greedy Equilibrium} (\GE)~\cite{L12}, if no agent can buy, swap or delete exactly one own edge to strictly decrease her cost.
These solution concepts induce the following $k$-local solution concepts:
$(G,\alpha)$ is in \emph{$k$-local Nash Equilibrium} (\kNE) if no agent can improve by a $k$-local move, $(G,\alpha)$ is in \emph{$k$-local Greedy Equilibrium} (\kGE) if no agent can improve by a $k$-local greedy move.
By slightly abusing notation, we will use the name of the above solution concepts to also denote the set of all instances which satisfy the respective solution concept, i.e., \NE denotes the set of all networks $(G,\alpha)$ that are in Pure Nash Equilibrium.
The notions \GE, \kNE, and \kGE are used respectively.
With this, we have the following:
\begin{observation}
\label{observation:NEinclusions}
For $k\geq 1$: $\NE \subseteq \kNE \subseteq \kGE$ and $\NE \subseteq \GE \subseteq \kGE$.
\end{observation}
We will also consider approximate equilibria and say that a network $(G,\alpha)$ is in \emph{$\beta$-approximate Nash Equilibrium}, if no strategy-change of an agent can decrease her cost to less than a $\beta$-fraction of her current cost in $(G,\alpha)$.
Similarly, we say $(G,\alpha)$ is in \emph{$\beta$-approximate Greedy Equilibrium}, if no agent can decrease her cost to less than a $\beta$-fraction of her current cost by buying, deleting or swapping exactly one own edge.

The \emph{social cost} of a network $(G_{\mathcal{S}},\alpha)$ is $\cost(G_{\mathcal{S}},\alpha) = \sum_{u\in V} \cost_u(G_\mathcal{S})$.
Let $\OPT_n$ be the minimum social cost of an $n$ agent network.
Let $\maxNE_n$ be the maximum social cost of any \NE network on $n$ agents and let $\minNE_n$ be the minimum social cost of any \NE network on $n$ agents.
Then, the \emph{Price of Anarchy} (PoA)~\cite{KP99} is the maximum over all $n$ of the ratio $\tfrac{\maxNE_n}{\text{OPT}_n}$, whereas the \emph{Price of Stability} (PoS)~\cite{ADKTWR} is the maximum over all $n$ of the ratio $\tfrac{\minNE_n}{\OPT_n}$.

\paragraph{Known Results.}
\textsc{Sum}-NCGs were introduced in \cite{Fab03}, where the authors proved the first PoA upper bounds, among them a constant bound for $\alpha \geq n^2$ and for trees.
Later, by different authors and papers \cite{Al06,De07,MS10,Mih13}, the initial PoA bounds were improved for several ranges of $\alpha$, resulting in the currently best known bounds of $\PoA=\mathcal{O}(1)$, in the range of $\alpha=\mathcal{O}(n^{1-\varepsilon})$, for any fixed $\varepsilon \geq \tfrac{1}{\log n}$, and the range of $\alpha \geq 65n$, as well as the bound $o(n^{\epsilon})$, for $\alpha $ between $\Omega(n)$ and $\mathcal{O}(n\log n)$.
Fabrikant et al.\ \cite{Fab03} also showed that for $\alpha < 2$ the network having minimum social cost is the complete network and for $\alpha \geq 2$ it is the spanning star which yields a constant PoS for all $\alpha$.
Since it is known that computing the optimal strategy change is NP-hard~\cite{Fab03}, Lenzner~\cite{L12} studied the effect of allowing only single buy/delete/swap operations, leading to efficiently computable best responses and $3$-approximate \NE{}s.
NCG versions where the cost of edges can be shared have been studied by Corbo \& Parkes~\cite{CP05} and Albers et al.~\cite{Al06}.

For the \textsc{Max}-NCG, Demaine et al.\ \cite{De07} showed that the PoA is at most $2$ for $\alpha\geq n$, for $\alpha$ in range $2\sqrt{\log n}\leq \alpha\leq n$ it is $\mathcal{O}(\min\{4^{\sqrt{\log n}},(n/\alpha)^{1/3}\})$, and $\mathcal{O}(n^{2/\alpha})$ for $\alpha < 2\sqrt{\log n}$.
For $\alpha>129$, Mihal\'{a}k \& Schlegel~\cite{MS10} showed, similarly to the \textsc{Sum}-NCG, that all equilibria are trees and the PoA is constant.

Kawald \& Lenzner~\cite{L11,KL13} studied convergence properties of the sequential versions of many NCG-variants and provided mostly negative results. The agents in these variants are myopic in the sense that they only optimize their next step which is orthogonal to our locality constraint.

To the best of our knowledge, the only models in the realm of NCGs that consider locality constraints are \cite{Bil14local} and \cite{Bil14traceroute}.
As discussed above, both model the local knowledge in a very pessimistic way.
Hence, it is not surprising that in \cite{Bil14local} the authors lower bound the PoA by $\Omega(n/(1+\alpha))$ for \textsc{Max}-NCG and by $\Omega(n/k)$ for \textsc{Sum}-NCG, when $k=o(\sqrt[3]{\alpha})$.
In particular, for \textsc{Max}-NCG they show that their lower bound is still $\Omega(n^{1-\varepsilon})$ for every $\varepsilon > 0$, even if $k$ is poly-logarithmic and $\alpha = \mathcal{O}(\log n)$.
On the bright side, they provide several PoA upper bounds that match with their lower bounds for different parameter combinations.
In their follow-up paper \cite{Bil14traceroute}, they equip agents with knowledge about either all distances to other nodes, a shortest path tree, or a set of all shortest path trees.
However, the \textsc{Max}-NCG PoA bounds are still $\Theta(n)$ for $\alpha > 1$, while the bounds for \textsc{Sum}-NCG improve to $\Theta(\min\{1+\alpha,n\})$. Note, that this is in stark contrast to the known upper bounds for the non-local version.

Apart from NCGs the influence of locality has already been studied for different game-theoretic settings, e.g. in the local matching model by Hoefer~\cite{H11}, where agents only know their $2$-neighborhood and choose their matching partners from this set.

\paragraph{Our Contribution.}
We introduce a new locality model for NCGs which allows to explore the intrinsic limits induced by a locality constraint and apply it to one of the mostly studied NCG versions, namely the \textsc{Sum}-NCG introduced by Fabrikant et al.~\cite{Fab03}.
In Section~\ref{section_hardness_dynamics}, we prove the counter-intuitive facts that constraining the agents' actions to their $k$-neighborhood has no influence on the hardness of computing good strategies and on the game dynamics even for very small $k$. In Section~\ref{section_approx}, we explore the impact of locality from the agents' perspective by studying the strength of local versus global strategy-changes obtaining an almost tight general approximation gap of $\Omega(\frac{\log n}{k})$, which is tight for trees. Finally, in Section~\ref{section:quality-of-equilibria} we provide drastically improved PoA upper bounds compared to \cite{Bil14traceroute}, which are in line with the known results about the non-local \textsc{Sum}-NCG. In contrast to this, we also prove a non-constant lower bound on the PoA, which proves that even in the most optimistic locality model the social efficiency deteriorates significantly.

\section{Computational Hardness and Game Dynamics}\label{section_hardness_dynamics}
In this section, we study the effect of restricting agents to $k$-local (greedy) moves on the hardness of computing best response and the convergence to equilibria.
We start with the observation that for any $k\geq 1$ computing a best possible $k$-local move is not easier than in the general setting. See Appendix~\ref{section:appendix_hardness} for omitted proofs.
\begin{theorem}\label{thm_br_computation}
    Computing a best possible $k$-local move is NP-hard for all $k\geq 1$.
\end{theorem}
\noindent Clearly, the best possible $k$-local greedy move of an agent can be computed in polynomial time by simply trying all possibilities. Similarly to~\cite{L12}, it is true that the best $k$-local greedy move is a $3$-approximation of the best $k$-local move. This yields:
\begin{theorem}\label{thm_3GE_vs_3NE}
    For any $k\geq 1$, every network in \kGE is in 3-approximate $k$-NE.
\end{theorem}
\noindent The same construction as in~\cite{L12} yields:
\begin{corollary}
For $k\geq 2$ there exist \kGE networks which are in $\tfrac{3}{2}$-approximate \kNE.
\end{corollary}
\noindent In the following, we analyze the influence of $k$-locality on the dynamics of the network creation process.
For this, we consider several sequential versions of the \textsc{Sum}-NCG, which were introduced in \cite{L11,KL13}, and refer to the corresponding papers for further details.
Our results only cover the \textsc{Sum}-NCG but we suspect similar results for corresponding versions of \textsc{Max}-NCG, which would be in line with the non-local version~\cite{KL13}.

In short, we consider the following network creation process:
Starting with any connected network having $n$ agents, edge price $\alpha$, and arbitrary edge ownership.
Now agents move sequentially, that is, at any time exactly one agent is active and checks if an improvement of her current cost is possible.
If this is the case, then this agent will perform the strategy-change towards her best possible new strategy. After every such move this process is iterated until there is no agent who can improve by changing her current strategy, i.e., the network is in equilibrium.
If there is a cyclic sequence of such moves, then clearly this process is not guaranteed to stop and we call such a situation a \emph{best response cycle} (BR-cycle). Note that the existence of a BR-cycle, even if the participating active agents are chosen by an adversary, implies the strong negative result that no ordinal potential function~\cite{MS96} can exist.  

In Table~\ref{table_dynamics}, we summarize our results for four types of possible strategy-changes: If agents can only swap own edges in their $k$-neighborhood, then this is called the \emph{$k$-local Asymmetric Swap Game} ($k$-ASG).
If agents are allowed to swap any incident edge within their $k$-neighborhood, then we have the \emph{$k$-local Swap Game} ($k$-SG).
If agents are allowed to perform $k$-local greedy moves only, then we have the \emph{$k$-local Greedy Buy Game} ($k$-GBG), and if any $k$-local move is allowed, then we have the \emph{$k$-local Buy Game} ($k$-BG).

\begin{table}
\centering
 \begin{tabular}{l|c|c}
 $k$ & \textsc{Sum} $k$-SG & \textsc{Sum} $k$-ASG  \\
 \hline
 $k=1$ & no moves [Thm.~\ref{thm_k=1_dynamics}] & no moves [Thm.~\ref{thm_k=1_dynamics}] \\
 \hline
 $k=2$ & BR-cycle [Thm.~\ref{thm_k=2_dynamics}] & OPEN \\
 \hline
 $k\geq 3$ & BR-cycle [Thm.~\ref{thm_k=2_dynamics}] & BR-cycle [Thm.~\ref{thm_k=3_dynamics}] \\
 \hline
 on trees & $\mathcal{O}(n^3)$ moves [Thm.~\ref{thm_tree_swap_dynamics}] & $\mathcal{O}(n^3)$ moves [Thm.~\ref{thm_tree_swap_dynamics}] 
 \end{tabular}
 ~\vspace*{0.3cm}\\
 \begin{tabular}{l|c|c}
 $k$  & \textsc{Sum} $k$-GBG & \textsc{Sum} $k$-BG \\
 \hline
 $k=1$  & $\Theta(n^2)$ moves [Thm.\ref{thm_k=1_dynamics}] & $\Theta(n)$ moves [Thm.~\ref{thm_k=1_dynamics}] \\
 \hline
 $k=2$ & BR-cycle [Thm.~\ref{thm_k=2_dynamics}] & BR-cycle [Thm.~\ref{thm_k=2_dynamics}]\\
 \hline
 $k\geq 3$  & BR-cycle [Thm.~\ref{thm_k=3_dynamics}] & BR-cycle [Thm.~\ref{thm_k=3_dynamics}]\\
 \hline
 on trees  & BR-cycle for $k\geq 2$ [Thm.~\ref{thm_k=2_dynamics}\&\ref{thm_k=3_dynamics}] & BR-cycle for $k\geq 2$ [Thm.~\ref{thm_k=2_dynamics}\&\ref{thm_k=3_dynamics}]
 \end{tabular}
 \caption{Overview of convergence speeds in the sequential versions. Note that the existence of a BR-cycle implies that there is no convergence guarantee. Here a move is a strategy-change by one agent.(For details see Appendix~\ref{section:appendix_dynamics}.)}
 \label{table_dynamics}
\end{table}

\section{Local versus Global Moves}\label{section_approx}
In this section, we investigate the agents' perspective.
We ask how much agents lose by being forced to act locally.
That is, we compare $k$-local moves to arbitrary non-local strategy-changes.
We have already shown that the best possible $k$-local greedy move is a $3$-approximation of the best possible $k$-local move.
The same was shown to be true for greedy moves versus arbitrary strategy-changes~\cite{L12}.
Thus, if we ignore small constant factors, it suffices to compare $k$-local greedy moves with arbitrary greedy moves. All omitted details of this section can be found in Appendix~\ref{section:appendix_approx}.

We start by providing a high lower bound on the approximation ratio for $k$-local greedy moves versus arbitrary greedy moves.
Corresponding to this lower bound, we provide two different upper bounds.
The first one is a tight upper bound for tree networks.
The second upper bound holds for arbitrary graphs, but is only tight for constant $k$.
Here, the structural difference between trees and arbitrary graphs is captured in the difference of Lemma~\ref{lem_tree_swap} and Lemma~\ref{lem_non-tree_swap}.
Whereas for tree networks only edge-purchases have to be considered, for arbitrary networks edge-swaps are critical as well.

\begin{theorem}[Locality Lower Bound]
    \label{thm_kGe_vs_Ge_approx_lower_bound}
    For any $n'$ there exist tree networks on $n\geq n'$ vertices having diameter $\Theta(\log n)$ which are in \kGE but only in $\Omega\left(\frac{\log n}{k}\right)$-approximate \GE.
\end{theorem}
\noindent For the first upper bound, let $(T,\alpha)$ be a tree network in \kGE and assume that it is not in \GE, that is, we assume there is an agent $u$, who can decrease her cost by buying or swapping an edge.
Note that we can rule out single edge-deletions, since these are $1$-local greedy moves and assume that no improving $k$-local greedy move is possible for any agent.
First, we will show that we only have to analyze the cost decrease achieved by single edge-purchases. 
\begin{lemma}\label{lem_tree_swap}
Let $T$ be any tree network. If an agent can decrease her cost by performing any single edge-swap in $T$, then this agent also can decrease her cost by performing a $k$-local single edge-swap in $T$, for any $k\geq 2$.
\end{lemma}
\noindent Now we analyze the maximum cost decrease achievable by buying one edge in $T$. We show, that our obtained lower bound from Theorem~\ref{thm_kGe_vs_Ge_approx_lower_bound} is actually tight. 
\begin{theorem}[Tree Network Locality Upper Bound]
\label{thm_tree_approx_upperbound}
Any tree network on $n$ vertices in \kGE is in $\mathcal{O}\left(\frac{\log n}{k}\right)$-approximate \GE.
\end{theorem}
\noindent We now prove a slightly inferior approximation upper bound for arbitrary networks. 
For this, we show that for any $k$ the diameter of any network is an upper bound on the approximation ratio. For constant $k$ this matches the $\Omega(\diam(G)/k)$ lower bound from Theorem~\ref{thm_kGe_vs_Ge_approx_lower_bound} and is almost tight otherwise.

\begin{lemma}\label{lem_non-tree_swap}
    For any $k\geq 2$ there is a non-tree network in which some agent cannot improve by performing a $k$-local swap but by a $k+1$-local swap.
\end{lemma}
\noindent It turns out that both buying and swapping operations yield the same upper bound on the approximation ratio, which is even independent of $k$.
\begin{theorem}[General Locality Upper Bound]
\label{thm_approx_upperbound}
    Let $G'$ be the network after some agent $u$ has performed a single improving edge-swap or -purchase in a network $G$, then
    $\frac{\cost_u(G)}{\cost_u(G')}\leq \diam(G).$ Thus, for any $k$ any network $G$ in \kGE is in $\mathcal{O}(\diam(G))$-approximate \GE.
\end{theorem}

\section{The Quality of $k$-Local Equilibria}
\label{section:quality-of-equilibria}
In this section we prove bounds on the Price of Anarchy and on the diameter of equilibrium networks. All omitted proofs can be found in Appendix~\ref{section:appendix_quality-proofs} and \ref{section:kne-and-ne-conformity}.

We start with a theorem about the relationship of $\beta$-approximate Nash equilibria and the social cost. This may be of independent interest, since it yields that the PoA can be bounded by using upper bounds on approximate best responses.
In fact, we generalize an argument by Albers et al.~(\cite{Al06}, proof of Theorem~3.2).
\begin{theorem}\label{thm_approx_diam_connection}
If for $\alpha\geq 2$, $(G,\alpha)$ is in $\beta$-approximate \NE,
then the ratio of social cost and social optimum is at most $\beta(3 + \diam(G))$.
\end{theorem}
\noindent However, for our bounds we will use another theorem, which directly carries over from the original non-local model.
\begin{theorem}[\cite{Fab03},\cite{AGT} Lemma 19.4]\label{thm_diam_PoA}
For any $k\geq 1$ and $\alpha \geq 2$, if a \kNE network $(G,\alpha)$ has diameter $D$, then its social cost is $\mathcal{O}(D)$ times the optimum social cost.
\end{theorem}
\noindent We obtain several PoA bounds by applying Theorem~\ref{thm_diam_PoA} and the diameter bounds shown later in this section.
Summarizing them, the next theorem provides our main insight into the PoA: The larger $k$ grows, the better the bounds get, providing a remarkable gap between $k=1$ and $k > 1$.
In other words, the theorem tells the ``price'' in terms of minimal required $k$-locality to obtain a specific PoA upper bound.
\begin{theorem}
\label{thm_collectionOfPoaUpperBounds}
For \kNE with $k\geq 1$ the following PoA upper bounds hold (cf.\ Fig.~\ref{fig:PoA_overview}):
\begin{description}
    \item [$k=1:$] $\PoA=\Theta(n)$, which can be seen by considering a line graph.
    \item [$k\geq 2:$] Theorem~\ref{thm:greedy_diameter} gives
        $
        \PoA =
        \begin{cases}
            \mathcal{O}((\alpha/k) + k)    & k < 2\sqrt{\alpha},\\
            \mathcal{O}(\sqrt{\alpha})                      & k \geq 2\sqrt{\alpha}.
        \end{cases}
        $
    \item [$k \geq 6:$] Theorem~\ref{thm_non_tree_diameter_upper_bound} gives $\PoA = \mathcal{O}\left(n^{1 - \varepsilon(\log(k-3) - 1))}\right)$, provided $1\leq \alpha \leq n^{1-\varepsilon}$ with $\varepsilon \geq 1/\log n$.
    Note, this yields a constant PoA bound for $\varepsilon \geq 1/(\log(k-3) - 1)$.
\end{description}
For tree networks and $k\geq 2$, Corollary~\ref{cor_tree_poa_upper_bound} gives $\PoA=\mathcal{O}\left(\log n\right)$.
\end{theorem}

\subsection{Conformity of $k$-Local and Nash Equilibria}
\label{subsection:eq-conformity}
We stated in Observation~\ref{observation:NEinclusions} that $\kNE \subseteq \NE$ holds.
In the following, we discuss the limits of known proof techniques to identify the parameters for that both equilibria concepts coincide, i.e., $\kNE = \NE$.
For proofs see Appendix~\ref{section:kne-and-ne-conformity}.

\begin{theorem}
\label{thm:whenkNEequalsNE}
The equilibrium concepts \kNE and \NE coincide for the following parameter combinations and yield the corresponding PoA results (cf.\ Fig.~\ref{fig:PoA_overview}):
\begin{enumerate}
    \item For $0< \alpha < 1$ and $2\leq k$, $\kNE = \NE$ and $\PoA=\mathcal{O}(1)$.
    \item For $1\leq \alpha \leq \sqrt{n/2}$ and $6\leq k$, $\kNE = \NE$ and $\PoA=\mathcal{O}(1)$.
    \item For $1\leq \alpha \leq n^{1-\varepsilon}$, $\varepsilon \geq 1/\log(n)$ and $4.667\cdot 3^{\lceil 1/\varepsilon\rceil} + 8\leq k$, $\kNE = \NE$ and $\PoA=\mathcal{O}(3^{\lceil 1/\varepsilon\rceil})$.
    \item For $1\leq \alpha \leq 12n\log n$ and $2\cdot 5^{1+\sqrt{\log n}} + 24\log(n) + 3 \leq k$, $\kNE = \NE$ and $\PoA=\mathcal{O}(5^{\sqrt{\log n}}\log n)$.
    \item For $12n\log n \leq \alpha$ and $2 \leq k$, $\kNE = \NE$ and $\PoA=\mathcal{O}(1)$.
\end{enumerate}
\end{theorem}
\begin{figure}
\centering
\includegraphics[width=.8\textwidth]{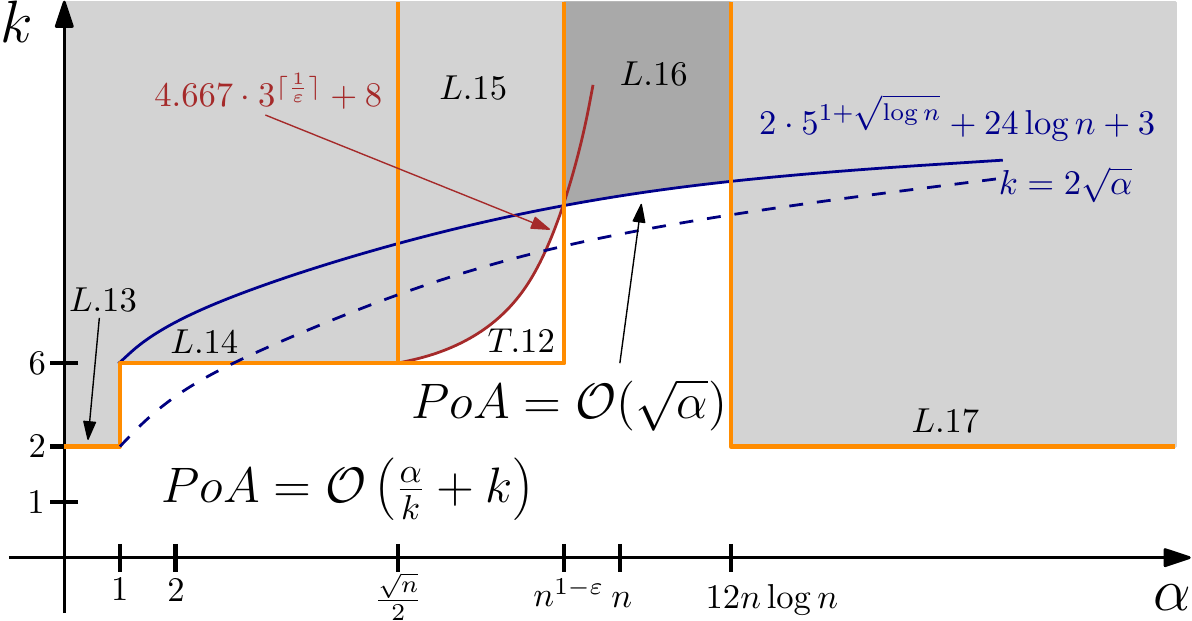}\vspace*{-0.4cm}
\caption{Overview of our results for both parameters $k$ and $\alpha$. The shaded areas show where \NE and \kNE coincide. In the light gray areas we have constant PoA, whereas in the dark gray area the PoA is bounded by $\mathcal{O}(5^{\sqrt{\log n}}\log n)$. (L=Lemma, T=Theorem)}
\label{fig:PoA_overview}
\end{figure}

\subsection{The Price of Anarchy for Tree Networks}
We consider the quality of tree networks that are in \kNE. First of all, by Lemma~\ref{lem_tree_swap}, we have that any tree network in \kGE must be in ASE and therefore, as shown in~\cite{Ehs11,MS12}, has diameter $\mathcal{O}(\log n)$. Together with Theorem~\ref{thm_diam_PoA} this yields the following upper bound.
\begin{corollary}
\label{cor_tree_poa_upper_bound}
    For \kNE tree networks with $2 \leq k \leq \log n$: $\PoA=\mathcal{O}(\log n)$.
\end{corollary}
\noindent Next, we show a non-constant lower bound which is tight if $k\in \mathcal{O}(1)$ or if $k\in \Omega(\log n)$.
The latter is true, since trees in \kNE have diameter $\mathcal{O}(\log n)$ and thus the technique in~\cite{Fab03} yields constant PoA.
In particular, this bound states that even with the most optimistic model of locality, there still exists a fundamental difference in the efficiency of the local versus the non-local NCG.
\begin{theorem}
    For \kNE tree networks with $2 \leq k \leq \log n$: $\PoA=\Omega\left(\frac{\log n}{k}\right)$.
\end{theorem}
\begin{proof}
 For any $2\leq k \leq d$ consider a complete binary tree $T_d$ of depth $d$ with root $r$, where every edge is owned by the incident agent which is closer to the root $r$ of $T$. Clearly, independent of $\alpha$ no agent in $T_d$ can delete or swap edges to improve. Thus, towards proving that $T_d$ is in $k$-NE, we only have to choose $\alpha$ high enough so that no agent can improve by buying any number of edges in her $k$-neighborhood. Any edge-purchase within the $k$-neighborhood of some agent $u$ can decrease $u$'s distances to other agents by at most $k-1$. Thus, if $\alpha = (k-1)n$, then no agent can buy one single edge to decrease her cost. Moreover, by buying more than one edge, agent $u$ cannot decrease her distance cost by more than $(k-1)n$, which implies that $(T_d,\alpha)$ is in $k$-NE for $\alpha = (k-1)n$.
 
 Now we consider the social cost ratio of $T_d$ and the spanning star OPT on $n = 2^{d+1}-1$ agents, which is the social cost minimizing network, since $\alpha \geq 2$. We will use $\cost(T_d) \geq \alpha(n-1) + n \dist_r(T_d)$, which is true, since the root $r$ has minimum distance cost among all agents in $T_d$. Note that $\dist_r(T_d) > \frac{n}{4} \log n$, since $r$ has at least $\frac{n}{2}$ many agents in distance $\frac{\log n}{2}$. Moreover, the spanning star on $n$ vertices has a total distance cost of less than $2n^2$, since it has diameter $2$. With $\alpha = (k-1)n$ this yields:
\begin{align*}
\frac{\cost(T_d)}{\cost(\text{OPT})} & > \frac{(k-1)\frac{n^2}{4} + \frac{n^2}{4}\log n}{(k-1)n^2 + 2n^2} > \frac{\frac{n^2}{4}((k-1)+\log n)}{n^2 (k+1)} \in \Omega\left(\frac{\log n}{k}\right).
\end{align*}
\end{proof}

\subsection{The Price of Anarchy for Non-Tree Networks}
We first provide a simple diameter upper bound that holds for any combination of $k\geq 2$ and $\alpha>0$.
For the omitted proofs see Appendix~\ref{section:kne-and-ne-conformity}.
\begin{theorem}
\label{thm:greedy_diameter}
Given a \kNE network $(G,\alpha)$ for $k\geq 2$, then it holds:
$$
\diam(G) \leq
\begin{cases}
    \alpha/(k-1) + k\frac{3}{2} + 1    & k < 2\sqrt{\alpha},\\
    2\sqrt{\alpha}                      & k \geq 2\sqrt{\alpha}.
\end{cases}
$$
\end{theorem}
\noindent In the following, we provide a more involved upper bound for $1\leq \alpha < n^{1-\varepsilon}$ by modifying an approach by \cite{De07}.
First, we give three lemmas that lower bound the number of agents in specific sized neighborhoods for \kNE networks.
Then, considering the locality parameter $k$, we look at the maximal neighborhoods for which these lower bounds apply and present an estimation on the network's diameter.

First, we restate Lemma~3 from \cite{De07}, which holds for any $k\geq 2$ since only $2$-local operations are considered.
\begin{lemma}[Lemma~3 from \cite{De07}]
\label{lemma:N2ballElements}
For any $k\geq 2$ and any \kNE network $G$ with $\alpha \geq 0$ it holds $|N_2(v)| > n/(2\alpha)$ for every agent $v\in V$.
\end{lemma}

\begin{lemma}
\label{lemma:ballSizeIncrease}
For $k\geq 6$, let $G$ be a \kNE network and $d\leq k/3 - 1$ an integer.
If there is a constant $\lambda > 0$ such that $|N_d(u)| > \lambda$ holds for every $u\in V$, then
\begin{itemize}
    \item[(1)] either $|N_{2d+3}(u)| > n/2$ for \emph{some} agent $u\in V$
    \item[(2)] or $|N_{3d+3}(v)| > \lambda \frac{n}{\alpha}$ for \emph{every} agent $v\in V$.
\end{itemize}
\end{lemma}

\begin{lemma}
\label{lemma:halfNodesToAllBall}
For $k\geq 4$, let $G$ be a \kNE network with $\alpha < n/2$ and $d\leq k/2 - 1$ an integer.
If there is an agent $u\in V$ with $|N_d(u)| \geq n/2$, then $|N_{2d+1}(u)| \geq n$.
\end{lemma}

\begin{theorem}
\label{thm_non_tree_diameter_upper_bound}
For $k \geq 6$, $n\geq 4$, and $1\leq \alpha \leq n^{1-\varepsilon}$ with $\varepsilon \geq 1/\log n$, the maximal diameter of any \kNE network is
$
    \mathcal{O}\left(n^{1 - \varepsilon(\log(k-3) - 1)}\right)
    .
$
\end{theorem}
\begin{proof}
Let $G$ be a \kNE network.
We define a sequence $(a_i)_{i\in\mathbb{N}}$ by $a_1\coloneqq 2$ and for any $i\geq 2$ with $a_i\coloneqq 3 a_{i-1} + 3$.
We want to apply Lemma~\ref{lemma:ballSizeIncrease} iteratively with $\lambda_i\coloneqq (n/\alpha)^i/2$.
Lemma~\ref{lemma:N2ballElements} ensures that with $|N_2(v)| > n/(2\alpha)=\lambda_1$ for all $v\in V$, we have a start for this.

Let $m$ be the highest sequence index with $a_m \leq k/4 - 3$.
If there is a $j \leq m$ such that case (1) of Lemma~\ref{lemma:ballSizeIncrease} applies, then there is a agent $u\in V$ with $|N_{2a_m + 3}(u)| > n/2$.
With $\alpha \leq n^{1-\varepsilon} < n/2$
we get with Lemma~\ref{lemma:halfNodesToAllBall} that $|N_{4 a_m + 7}| \geq n$ holds and hence the diameter is at most $a_m < k$.
Else, case (2) applies for all $i\leq m$ and we know that for every $v\in V$ it holds $|N_{3a_m+3}(v)| > (n/\alpha)^{m-1} / 2$.
Using $a_i = \frac{7}{6} 3^i - 3/2$ and $a_m \leq k/4 - 3$, we get $m\geq \log(k-3) - 1$.

Let $D$ be the diameter of $G$ and $p$ a longest shortest path.
We define a set $C$ by selecting the first agent of $p$ as $c_1$ and than along the path selecting every further agent with distance of $2k$ to the last selected agent.
Now consider the operation of $c_1$ buying an edge to a agent at distance $k$ in the direction of $c_2$.
Using $k\geq 3a_m + 3$, $|N_{3a_m + 3}(c)| > (n/\alpha)^{m-1} /2$ for all $c\in C$ and that $G$ forms an equilibrium:
$$
    \alpha \geq (k-1)(|C|-1)(n/\alpha)^{\log(k-3)-2} /2
         \geq \frac{k-1}{2}\left(\frac{D}{2k} - 1\right)(n/\alpha)^{\log(k-3)-2}
$$
This gives:
$
    D \leq \frac{4k}{k-1}\alpha\left(\frac{\alpha}{n}\right)^{\log(k-3)-2} + 2k
      \leq 5 n^{1 - \varepsilon(\log(k-3) - 1)} + 2k
.
$
By using Lemma~\ref{lemma_poa_upper_bound_for_smaller_n}, we get the claimed diameter upper bound for any $k\geq 6$.
\end{proof}

\section{Conclusion and Open Problems}
Our results show a major gap in terms of social efficiency between the worst case locality model by Bilò et al.\ \cite{Bil14local} and our more optimistic locality assumption. This gap is to be expected since agents in our model can base their decisions on more information. Interestingly, since most of our upper bounds on the PoA are close to the non-local model, this shows that the natural approach of probing different local strategies is quite convenient for creating socially efficient networks. On the other hand, the non-constant lower bound on the PoA and our negative results concerning the approximation of non-local strategies by local strategies show that the locality constraint does have a significant impact on the game, even in the most optimistic setting. Moreover, our negative results on the hardness of computing best responses and on the dynamic properties show that these problems seem to be intrinsically hard and mostly independent of locality assumptions.

The exact choice of the distance cost function seems to have a strong impact in our model. For the \textsc{Max}-NCG, which was extensively studied by Bilò et al.\ \cite{Bil14local}, it is easy to see that a cycle of odd length, where every agent owns exactly one edge, yields even in our model a lower bound of $\Omega(n)$ on the PoA for $k=2$ and $\alpha > 1$. This seems to be an interesting contrast between the \textsc{Sum}-NCG and the \textsc{Max}-NCG, which should be further explored. 
It might also be interesting to further study what happens if agents are limited to probe only a certain number of strategies. So far, we only considered the cases of probing all strategies and of probing the quadratic number of greedy strategies.
Also, applying our locality approach to other models, in particular those motivated from the economics perspective, like \cite{jackson1996strategic} and \cite{BG00}, seems a natural next step of inquiry.

\bibliographystyle{abbrv}
\bibliography{locallocal}

\newcommand{\SortNoop}[1]{}
\begin{thebibliography}{10}

\bibitem{Al06}
S.~Albers, S.~Eilts, E.~Even{-}Dar, Y.~Mansour, and L.~Roditty.
\newblock On nash equilibria for a network creation game.
\newblock {\em {ACM} Trans. Economics and Comput.}, 2(1):2, 2014.

\bibitem{ADHL10}
N.~Alon, E.~D. Demaine, M.~T. Hajiaghayi, and T.~Leighton.
\newblock Basic network creation games.
\newblock {\em {SIAM} J. Discrete Math.}, 27(2):656--668, 2013.

\bibitem{ADKTWR}
E.~Anshelevich, A.~Dasgupta, J.~Kleinberg, E.~Tardos, T.~Wexler, and
  T.~Roughgarden.
\newblock The price of stability for network design with fair cost allocation.
\newblock {\em SIAM Journal on Computing}, 38(4):1602--1623, 2008.

\bibitem{flpLocalityGapArya}
V.~Arya, N.~Garg, R.~Khandekar, A.~Meyerson, K.~Munagala, and V.~Pandit.
\newblock Local search heuristics for k-median and facility location problems.
\newblock {\em {SIAM} J. Comput.}, 33(3):544--562, 2004.

\bibitem{BG00}
V.~Bala and S.~Goyal.
\newblock A noncooperative model of network formation.
\newblock {\em Econometrica}, 68(5):1181--1229, 2000.

\bibitem{Bil14local}
D.~Bil\`{o}, L.~Gual\`{a}, S.~Leucci, and G.~Proietti.
\newblock Locality-based network creation games.
\newblock In {\em {SPAA} 2014}, pages 277--286, New York, NY, USA, 2014. ACM.

\bibitem{Bil14traceroute}
D.~Bilò, L.~Gualà, S.~Leucci, and G.~Proietti.
\newblock Network creation games with traceroute-based strategies.
\newblock In {\em Structural Information and Communication Complexity}, volume
  8576 of {\em LNCS}, pages 210--223. Springer International Publishing, 2014.

\bibitem{CP05}
J.~Corbo and D.~Parkes.
\newblock The price of selfish behavior in bilateral network formation.
\newblock In {\em {PODC} 2005 Proceedings}, pages 99--107, New York, NY, USA,
  2005. ACM.

\bibitem{De07}
E.~D. Demaine, M.~T. Hajiaghayi, H.~Mahini, and M.~Zadimoghaddam.
\newblock The price of anarchy in network creation games.
\newblock {\em ACM Trans. on Algorithms}, 8(2):13, 2012.

\bibitem{Ehs11}
S.~Ehsani, M.~Fazli, A.~Mehrabian, S.~S. Sadeghabad, M.~Safari, M.~Saghafian,
  and S.~ShokatFadaee.
\newblock On a bounded budget network creation game.
\newblock In {\em {SPAA} 2011}, pages 207--214. {ACM}, 2011.

\bibitem{Fab03}
A.~Fabrikant, A.~Luthra, E.~Maneva, C.~H. Papadimitriou, and S.~Shenker.
\newblock On a network creation game.
\newblock In {\em {PODC} 2003 Proceedings}, pages 347--351. ACM, 2003.

\bibitem{H11}
M.~Hoefer.
\newblock Local matching dynamics in social networks.
\newblock In {\em Automata, Languages and Programming - 38th International
  Colloquium, {ICALP} 2011, Zurich, Switzerland, July 4-8, 2011, Proceedings,
  Part {II}}, pages 113--124, 2011.

\bibitem{jackson1996strategic}
M.~O. Jackson and A.~Wolinsky.
\newblock A strategic model of social and economic networks.
\newblock {\em Journal of economic theory}, 71(1):44--74, 1996.

\bibitem{KL13}
B.~Kawald and P.~Lenzner.
\newblock On dynamics in selfish network creation.
\newblock In {\em Proc. 25th ACM Symp. on Parallelism in Alg. and
  Architectures}, pages 83--92. ACM, 2013.

\bibitem{KP99}
E.~Koutsoupias and C.~Papadimitriou.
\newblock Worst-case equilibria.
\newblock In {\em {STACS} 1999}, pages 404--413, Berlin, Heidelberg, 1999.
  Springer-Verlag.

\bibitem{L11}
P.~Lenzner.
\newblock On dynamics in basic network creation games.
\newblock In G.~Persiano, editor, {\em Algorithmic Game Theory}, volume 6982 of
  {\em LNCS}, pages 254--265. Springer, 2011.

\bibitem{L12}
P.~Lenzner.
\newblock Greedy selfish network creation.
\newblock In P.~Goldberg, editor, {\em Internet and Network Economics}, LNCS,
  pages 142--155. Springer Berlin Heidelberg, 2012.

\bibitem{Mih13}
A.~Mamageishvili, M.~Mihalák, and D.~Müller.
\newblock Tree nash equilibria in the network creation game.
\newblock In A.~Bonato, M.~Mitzenmacher, and P.~Prałat, editors, {\em Alg. and
  Models for the Web Graph}, volume 8305 of {\em LNCS}, pages 118--129.
  Springer, 2013.

\bibitem{MS10}
M.~Mihal\'{a}k and J.~C. Schlegel.
\newblock The price of anarchy in network creation games is (mostly) constant.
\newblock In {\em {SAGT} 2010}, pages 276--287. Springer-Verlag, 2010.

\bibitem{MS12}
M.~Mihal\'{a}k and J.~C. Schlegel.
\newblock Asymmetric swap-equilibrium: A unifying equilibrium concept for
  network creation games.
\newblock In {\em {MFCS} 2012}, volume 7464 of {\em LNCS}, pages 693--704.
  Springer Berlin / Heidelberg, 2012.

\bibitem{MS96}
D.~Monderer and L.~S. Shapley.
\newblock Potential games.
\newblock {\em Games and Economic Behavior}, 14(1):124 -- 143, 1996.

\bibitem{AGT}
N.~Nisan, T.~Roughgarden, E.~Tardos, and V.~V. Vazirani.
\newblock {\em Algorithmic Game Theory}.
\newblock Cambridge University Press, New York, NY, USA, 2007.

\bibitem{WS11}
D.~P. Williamson and D.~B. Shmoys.
\newblock {\em The design of approximation algorithms}.
\newblock Cambridge University Press, 2011.

\end{thebibliography}

\newpage
\appendix

\section{Omitted Hardness Details from Section~\ref{section_hardness_dynamics}}
\label{section:appendix_hardness}
\begin{proof}[Proof of Theorem~\ref{thm_br_computation}]
 We reduce from the well-known \textsc{Dominating Set} problem~\cite{WS11}. Given any graph $G=(V,E)$, then the \textsc{Dominating Set} problem asks for a minimum cardinality dominating set $D \subseteq V$. A set $D$ is dominating for a graph $G$, if every vertex of $G$ belongs to $D$ or has a neighbor in $D$.

 Let $G = (V,E)$ be any \textsc{Dominating Set} instance. Let $(G',\alpha)$ be an instance of the NCG obtained from $G$ in the following way: network $G'$ is the network $G$ where the edge ownership of each edge is chosen arbitrarily among its end-points. Furthermore, in $G'$ there is a new agent $u \notin V$ who owns an edge towards all other agents. We choose the edge-price $\alpha$ such that $1 < \alpha < 2$ holds.

 We claim that agent $u$'s minimum cost strategy which $u$ can obtain by modifying owned edges in her $k$-neighborhood in $(G',\alpha)$ is identical to the minimum cardinality dominating set in $G$ for any $k\geq 1$.

 Let $S_u$ be the current strategy of $u$ in $(G',\alpha)$ and let $S_u^k$ be a minimum cost strategy agent for $u$ which can be obtained by deleting edges and by buying or swapping edges in the $k$-neighborhood of $u$ in $G'$. First of all, notice that agent $u$ has eccentricity $1$ in $G'$ and, since $u$ clearly does not want to buy additional edges, we have that any strategy-change for agent $u$ is performed by edge-deletions only. Deleting an edge is a modification in the $1$-neighborhood, thus we have that $S_u^1 = S_u^2 = \dots $ for all $k\geq 1$. Let $G''$ be the network obtained by agent $u$'s strategy-change towards $S_u^1$. Notice that in $G''$ agent $u$ cannot have eccentricity of at least $3$, since $1<\alpha<2$ and thus not deleting the edges to agents in distance at least $3$ in $G''$ would yield a strategy having strictly lower cost which contradicts our assumption that $S_u^1$ yields minimum cost for agent $u$. It follows that $S_u^1$ must be a dominating set in $G$. Moreover, since $\alpha >1$ and since $S_u^1$ is agent $u$'s minimum cost strategy in $(G',\alpha)$, it follows that $S_u^1$ must be a minimum cardinality dominating set, since any additional edge strictly increases agent $u$'s cost.
\end{proof}

\begin{proof}[Proof of Theorem~\ref{thm_3GE_vs_3NE}]
We show that if an agent cannot improve by a $k$-local greedy move, then this agent can decrease her cost to at most a third of her current cost by performing a $k$-local move. 

For this, similar to \cite{L12}, we reduce the best response computation of any agent to the solution of a corresponding \textsc{Uncapacitated Metric Facility Location} (UMFL)\cite{WS11} instance.
Here, UMFL is the problem to select a subset $X\subset\mathcal{F}$ of facilities, such that for a given set of clients, individual opening costs $f_v\geq 0$ for every facility $v\in\mathcal{F}$, and a metric distance function $d:\mathcal{F}\times \mathcal{C}\rightarrow \mathbb{R}$, such that $\sum_{v\in X} f_v + \sum_{x\in \mathcal{C}} \min_{v\in X}d(x,v)$ is minimized.
Arya et al.\ \cite{flpLocalityGapArya} provide a locality gap result that (beside other results) states:
When starting with an arbitrary facility set and performing only the operations of closing a single facility, opening a single facility, or swapping a facility until no further improvement is possible, this greedy local search heuristic results in a $3$-approximation of the optimal solution.

Given a \kGE network $G$ with nodes $V$, edges $E$ and edge creation cost $\alpha$, consider an arbitrary node $u\in V$.
For $u$ let $S_u$ be the set of agents to which $u$ owns an edge and let $\overline{S_u}$ be the set of agents which own edges to agent $u$.
Using this, we define a special instance $I=(\mathcal{F},\mathcal{C},\{f_v\},d)$ for the UMFL problem:
\begin{itemize}
    \item the set of facilities $\mathcal{F}$ is given by $\mathcal{F}\coloneqq N_k(u)\setminus \{u\}$
    \item the set of clients $\mathcal{C}$ is given by $\mathcal{C}\coloneqq V\setminus \{u\}$
    \item for every facility $v\in\mathcal{F}\cap \bar s_u$, we define the opening cost as $f_v\coloneqq 0$, and for all others as $f\coloneqq \alpha$
    \item for a facility $v\in\mathcal{F}$ and a client $x\in\mathcal{C}$ we set the distance as $d(v,x)\coloneqq d_G(v,x) + 1$; the distance is $\infty$, if there is no path from $v$ to $x$ in $G$
\end{itemize}
Note that by using the shortest path metric to define the distances in $I$, we ensure that the distances in $I$ are metric.
It is easy to see that $\cost_u(G) = \cost(I) = \sum_{v\in S_u} f_v + \sum_{x\in \mathcal{C}} \min_{v\in S_u}d(x,v)$.
Since  we assume that agent $u$ cannot perform any improve $k$-local greedy move, the locality gap for UMFL~\cite{flpLocalityGapArya} yields that the cost of agent $u$ in $G$ is at most $3$ times her cost after performing a best possible $k$-local move.
\end{proof}

\section{Detailed Discussion of Game Dynamics}\label{section:appendix_dynamics}
We start with the case $k=1$:
\begin{theorem}\label{thm_k=1_dynamics}
 If $k=1$, then no agent can perform any strategy-change if only edge-swaps are allowed or in all versions if starting network is a tree. In the $1$-GBG or the $1$-BG there is guaranteed convergence to a $1$-\GE or $1$-\NE in $\Theta(n^2)$ many moves.
\end{theorem}
\begin{proof}
 Clearly, if $k=1$, then no agent can create new edges, that is, edges to non-neighbors, since all non-neighbors are in distance at least $2$. Thus, in the $1$-SG or the $2$-SG no strategy-change is possible, which implies that any network (even disconnected networks) are in $1$-Swap Equilibrium\footnote{This is the corresponding $k$-local solution concept to the \emph{Swap Equilibrium} introduced by Alon et al.~\cite{ADHL10}.} and in $1$-ASE. Similarly, we have that any tree network is in $1$-\GE or in $1$-\NE.
 
 If we start with a non-tree network in the $1$-GBG or the $1$-BG, then agents can only delete edge to improve on their current situation. Thus, the number of edges in the network is an ordinal potential function~\cite{MS96} for the game, which implies guaranteed convergence. Since a $n$-agent network can have $\Theta(n^2)$ many edges and since the network creation process stops if a tree is created, it follows that there can be $\mathcal{O}(n^2)$ many moves in the $1$-GBG, since only single edges can be deleted in one move, and $\mathcal{O}(n)$ many moves in the $1$-BG, since an agent may remove all her edges in one move. For both versions, a complete network, where every agent owns roughly half her incident edges, and $\alpha > 2$ yields a matching lower bound on the number of moves, since the process must converge to a star. 
\end{proof}
\noindent For the case $k=2$ we have:
\begin{theorem}\label{thm_k=2_dynamics}
 If $k=2$, then there exists a best response cycle for the $2$-SG, the $2$-GBG and the $2$-BG. For the $2$-GBG and the $2$-BG such a cycle can be reached even if the starting network is a tree. 
\end{theorem}
\begin{proof}
 If we consider the $2$-SG, then we can simply use the best response cycle construction from \cite{L11}, since this only uses $2$-local edge-swaps. Actually, this proves that there is a BR-cycle in the $2$-SG for all $k\geq 2$, since in every step of the cycle these $2$-local edge-swaps are best possible swaps in the corresponding network.
 
 For the $2$-GBG and the $2$-BG Fig.~\ref{fig:2local_br_cycle} shows a best response cycle.
 \begin{figure}[h!]
  \centering
  \includegraphics[width=0.9\textwidth]{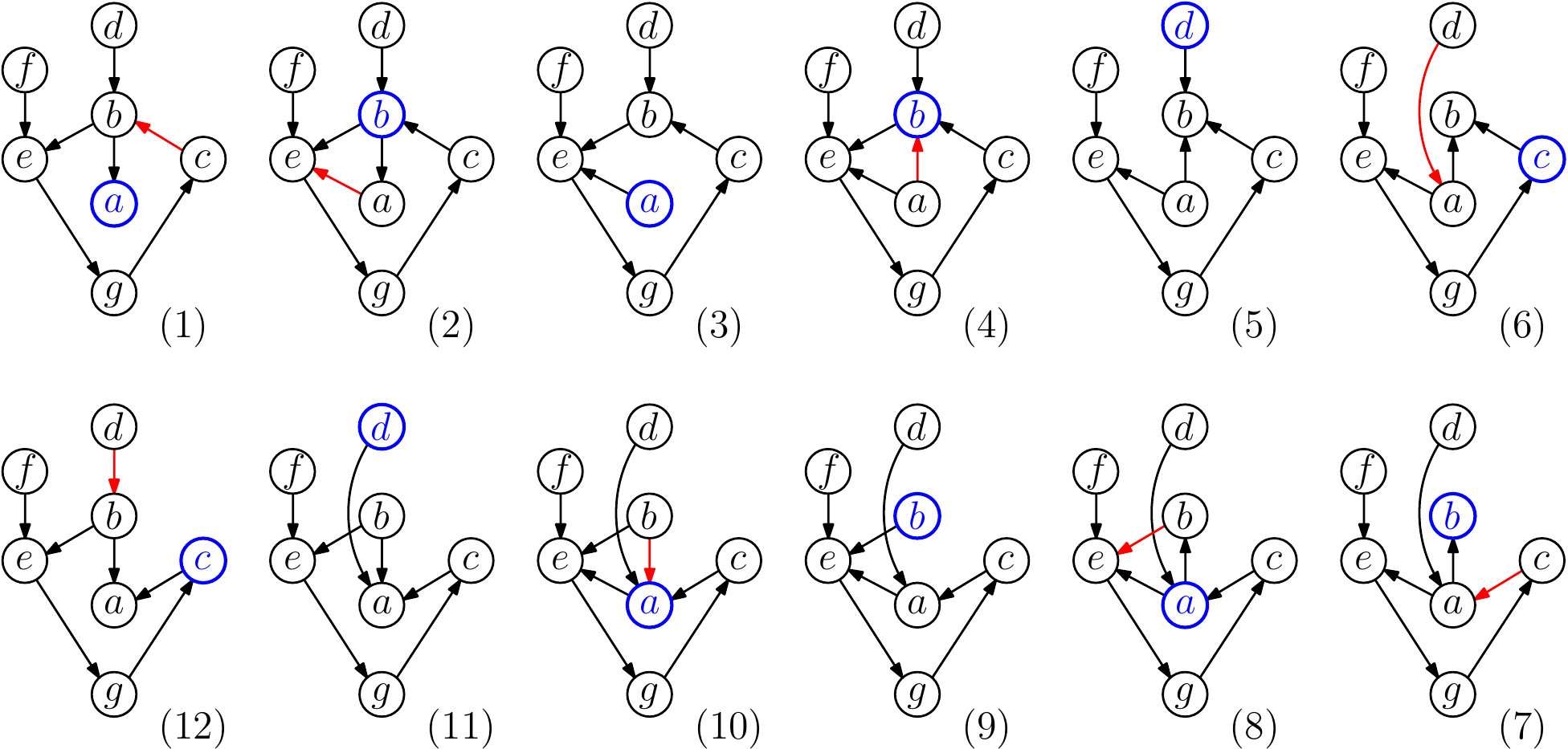}
  \caption{The BR-cycle for $k=2$ and $2< \alpha <3$. Network ($i$) leads to ($i+1 \bmod 12$) as follows: 
  In (1) $a$ buys $ae$, in (2) $b$ deletes $ab$, in (3) $a$ buys $ab$, in (4) $b$ deletes $be$, in (5) $d$ swaps $db\to da$, in (6) $c$ swaps $cb \to ca$, in (7) $b$ buys $be$, in (8) $a$ deletes $ab$, in (9) $b$ buys $ab$, in (10) $a$ deletes $ae$, in (11) $d$ swaps $da \to db$, in (12) $c$ swaps $ca \to cb$. Note that in any step of the cycle there may be more than one agent who can improve, e.g. agent $f$ can also improve by swapping $fe \to fb$ in network (1). The active agents are chosen by an adversary to sustain the BR-cycle.}
  \label{fig:2local_br_cycle}
 \end{figure}
 It is easy to check, that in every step of the cycle in Fig.~\ref{fig:2local_br_cycle} the active agent performs a strategy-change to her best possible strategy and that this strategy-change is a $2$-local greedy move. Observe, that in network (5) from Fig.~\ref{fig:2local_br_cycle} agent $d$ actually has a strictly better $3$-local greedy move (the swap $db \to de$), than her best possible $2$-local greedy move. Thus, the shown construction cannot be used for the case $k=3$.
 
 There actually is a best response path starting with a tree network towards which ends in network (1) of the BR-cycle in Fig.~\ref{fig:2local_br_cycle}. This shows, that there is no convergence guarantee even if the starting network is a tree. The best response path is depicted in Fig.~\ref{fig:2local_br_path}. 
 \begin{figure}[h!]
  \centering
  \includegraphics[width=0.4\textwidth]{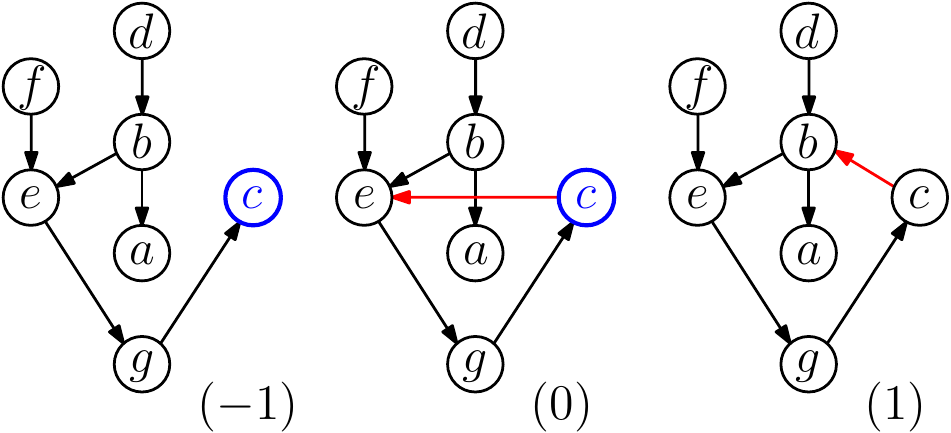}
  \caption{The BR-path for $k=2$ and $2< \alpha <3$ from a tree network to network (1) from Fig.~\ref{fig:2local_br_cycle}. 
  In (-1) $c$ buys $ce$, in (0) $c$ swaps $ce \to cb$.}
  \label{fig:2local_br_path}
 \end{figure}
 Again, it is easy to see that all strategy-changes shown in Fig.~\ref{fig:2local_br_path} are towards the best possible strategy in within the $2$-neighborhood and that all moves are $2$-local greedy moves.
\end{proof}
\noindent For the case $k\geq 3$, we have:
\begin{theorem}\label{thm_k=3_dynamics}
 For any $k\geq 3$ there exists a best response cycle for the $k$-ASG, the $k$-GBG and the $k$-BG.
\end{theorem}
\begin{proof}
 For the $k$-ASG with $k\geq 3$, we can simply reuse the corresponding construction in~\cite{KL13}. There the existence of a best response cycle is shown, where the best possible strategy-changes are in fact $3$-local greedy moves. 
 
 For the $3$-GBG and the $3$-BG Fig.~\ref{fig:3local_br_cycle} shows a best response cycle.
 \begin{figure}[h!]
  \centering
  \includegraphics[width=0.6\textwidth]{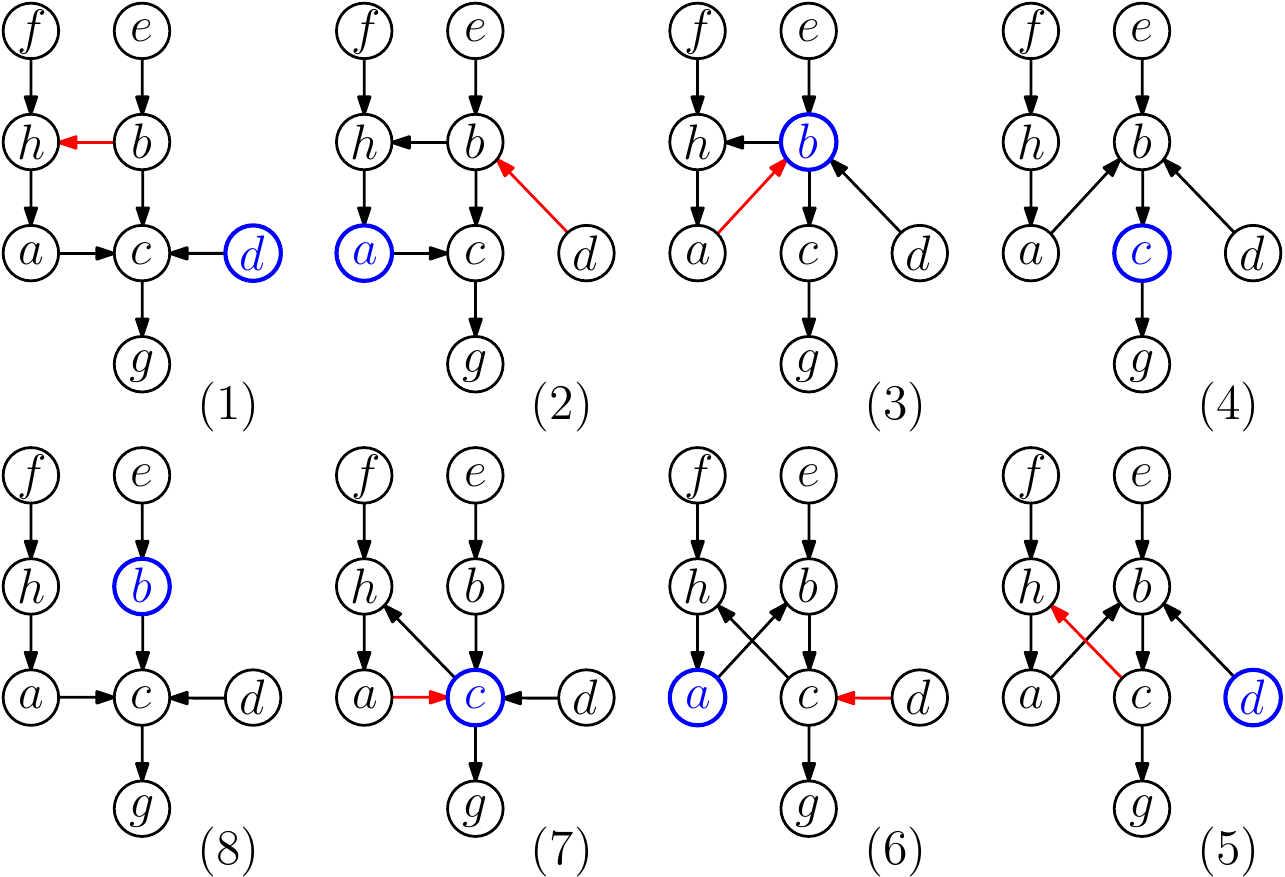}
  \caption{The BR-cycle for $k=3$ and $3< \alpha <4$. Network ($i$) leads to ($i+1 \bmod 8$) as follows: 
  In (1) $d$ swaps $dc \to db$, in (2) $a$ swaps $ac \to ab$, in (3) $b$ deletes $bh$, in (4) $c$ buys $ch$, in (5) $d$ swaps $db\to dc$, in (6) $a$ swaps $ab \to ac$, in (7) $c$ deletes $ch$, in (8) $b$ buys $bh$.}
  \label{fig:3local_br_cycle}
 \end{figure}
 It is easy to check, that in every step of the cycle in Fig.~\ref{fig:3local_br_cycle} the active agent performs a best possible strategy-change which happens to be a $k$-local greedy move. Moreover, note that the best response cycle contains trees, which implies that even with a tree network, there is no guaranteed convergence for $k=3$.
 
 For $k\geq 4$, we can reuse the corresponding constructions in \cite{KL13} to show that there exists a best response cycle in the $k$-GBG and the $k$-BG for $k\geq 4$.
\end{proof}
\noindent Finally, we analyze the convergence behavior if the starting network is a tree and agents are only allowed to swap edges.
\begin{theorem}\label{thm_tree_swap_dynamics}
 Starting from a tree network the network creation process converges in $\mathcal{O}(n^3)$ many moves to an equilibrium network in both the $k$-SG and the $k$-ASG.
\end{theorem}
\begin{proof}
 The statement follows, since it was shown in \cite{L11,KL13} that there exists an ordinal potential function for this scenario, even if agents can perform arbitrary swaps. Thus, this carries over to our setting, where only swaps within an agent's $k$-neighborhood are allowed. Moreover, it was shown in \cite{L11,KL13}, then there can be at most $\mathcal{O}(n^3)$ many moves until an equilibrium is reached. Restricting the number of strategy-changes of the agents can only decrease this number.
\end{proof}

\section{Omitted Details from Section~\ref{section_approx}}\label{section:appendix_approx}
We analyze the relationship between $k$-greedy stability and greedy stability. In $k$-greedy stable networks, no agent can improve by buying, deleting or swapping one own edge within her $k$-neighborhood. In contrast, in greedy-stable networks, no agent can improve by buying, deleting or swapping one own edge even if all possible moves are allowed.

For the following constructions, we have to analyze an agent's distance cost towards vertices in a special subnetwork.

Let the subtree $H_{d,l}$ of some tree network $G$ for some agent $u$ be defined as follows: $H_{d,l}$ is a complete balanced binary tree of depth $d$, where the edges are owned by the endpoint which is closer to the root, with an additional agent $v$ who has distance $l$ towards $u$ and who owns an edge towards the root of the tree.

Let $K$ be any induced subgraph of $G$, then $\dist_u(K) = \sum_{v\in V(K)} \dist_G(u,v)$ is the distance cost of agent $u$ towards all agents in the subnetwork $K$.
\begin{lemma}\label{lem_parttreecost}
 We have that $\dist_u(H_{d,l}) = 2^{d+1}(d+l)+1$.
\end{lemma}
\begin{proof}
 We prove the statement by induction on $l$. For $l=0$ we have \begin{align*}\dist_u(H_{d,0}) &= \sum_{i=0}^d 2^i(i+1) = \sum_{i=0}^d 2^i i + \sum_{i=0}^d 2^i = 2^{d+1}(d-1)+2 + 2^{d+1}-1 \\&= 2^{d+1}(d+0) + 1.\end{align*}
 For the induction step from $l$ to $l+1$ we use that $|V(H_{d,l})| = 2^{d+1}$ for any $l$. With this, the step be seen as follows:
 \begin{align*}
    \dist_u(H_{d,l+1}) &= \dist_u(H_{d,l}) + |V(H_{d,l})| = 2^{d+1}(d+l)+1 + 2^{d+1} \\&= 2^{d+1}(d+l+1) +1.
 \end{align*}
\end{proof}
\noindent Let $T_d$ be a complete balanced binary tree of depth $d$, let $u$ be a leaf of $T_d$, let $r$ be the root of $T_d$ and let $x$ be the neighbor of $r$ who has the largest distance towards $u$ in $T_d$. 
We have the following:
\begin{lemma}\label{lem_distcost_Td} 
Let $d$ be even, then we have $\dist_u(T_d) = 2^{d+1}(2d-3) + d + 6$.
\end{lemma}
\begin{proof}
 We compute the distance cost by rephrasing it as sum of distance cost towards special trees of the form $H_{d,l}$ and then we use Lemma~\ref{lem_parttreecost}.
 \begin{align*}
  \dist_u(T_d) &= \sum_{i=1}^d \dist_u(H_{i-1,i}) = \sum_{i=1}^d (2^{i}(2i-1)+1) \\
  &= 2\sum_{i=1}^d (2^{i}i) - \sum_{i=1}^d 2^i + d = 2(2^{d+1}(d-1)+2) - ( 2^{d+1}-2) + d\\
  &= 2^{d+1}(2d-3) + d + 6.
 \end{align*}
\end{proof}
\noindent Now we are ready to prove Theorem~\ref{thm_kGe_vs_Ge_approx_lower_bound}.
\begin{proof}[Proof of Theorem~\ref{thm_kGe_vs_Ge_approx_lower_bound}]
 We consider the $d$-$l$-Tree-Star, i.e., a combination of a complete binary tree with depth $d$ and a star having $l$ leaves. A $d$-$l$-Tree-Star $G_{d,l}$ is constructed as follows: Let $T_d$ be a complete binary tree where every edge is owned by the endpoint which is closer to the root of the tree and let $r$ be the root of $T_d$. Let $S_l$ be a star having $l$ leaves and let $z$ be the center vertex of the star owning all edges of $S_l$.
The network $G_{d,l}$ is obtained by adding an agent $y$, which owns an edge towards $r$ and an edge towards $z$ and thereby connects the subgraphs $T_d$ and $S_l$. See Fig.~\ref{fig:tree_star}.
\begin{figure}[ht]
 \centering
 \includegraphics[width=0.5\textwidth]{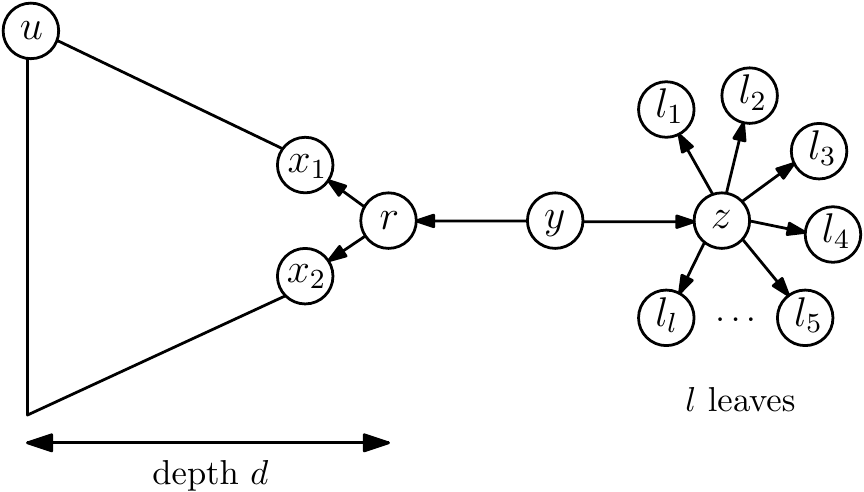}
 \caption{Illustration of the $d$-$l$-Tree-Star $G_{d,l}$.}
 \label{fig:tree_star}
\end{figure}

\noindent Clearly, we have that $|V(G_{d,l})| = 2^{d+1}+l+1$, so for any $n'$ we can choose $d$ and $l$ such that $|V(G_{d,l})| = n \geq n'$ holds. 

Observe that no agent can swap any edge to decrease her cost and, since $G_{d,l}$ is a tree, no agent can improve by deleting any edge. We will focus on agent $u$, that is, any leaf of the tree $T_d$. Note, that $u$ has eccentricity of $d+3$. For agent $u$, we have
\begin{equation}\dist_u(G_{d,l}) = \dist_u(T_d) + 2d+3 + l(d+3) = 2^{d+1}(2d-3)+3d+l(d+3)+9, \label{eq:cost_leaf}
 \end{equation}
where we have used Lemma~\ref{lem_distcost_Td}.

In the following, we will only consider $d$-$l$-Tree-Stars $G_{d,l}$ where $d$ is even. Let $G_{d,l}''$ be the network which is obtained from $G_{d,l}$ if agent $u$ buys the edge $uz$ in $G_{d,l}$. We will need agent $u$'s distance cost in $G_{d,l}''$.
\begin{lemma}\label{lem_newleafcost}
If $d$ is even then we have $$\dist_u(G_{d,l}'') = 2^{d+1}d + d -3\cdot 2^{\frac{d}{2}+2} + 2^{d+2} + 2l + 9 .$$
\end{lemma}
\begin{proof}
 We will prove the statement by using Lemma~\ref{lem_parttreecost}. We have
 \begin{align*}
  \dist_u(G_{d,l}'') &= 3 + 2l + \sum_{i=1}^{\frac{d}{2}+1}\dist_u(H_{i-1,i}) + \sum_{i=3}^{\frac{d}{2}+1}\dist_u(H_{d-i+2,i})\\
  &= 3+2l + \sum_{i=1}^{\frac{d}{2}+1}\left(2^i(2i-1)+1\right) + \sum_{i=3}^{\frac{d}{2}+1}\left(2^{d-i+3}(d+2)+1\right)\\
  &= 3 + 2l + 2^{d+1}d + d - 3\cdot 2^{\frac{d}{2}+2} + 2^{d+2} + 6.
 \end{align*}
\end{proof}
\noindent Towards a high lower bound for the approximation, we will choose $l$ large enough such that buying the edge $uz$ will be the best possible edge-purchase of agent $u$.

\begin{lemma}\label{lem_lowerbound_for_l}
 If $d$ is even and $l\geq 2^{d+1} - 2^{\frac{d}{2}+2}$  then buying the edge $uz$ is the best possible single edge-purchase of agent $u$.
\end{lemma}
\begin{proof}
 Let $G_{d,l}'$ be the network which is obtained if agent $u$ buys the edge $uy$ in network $G_{d,l}$. To prove the statement, it suffices to choose $l$ large enough such that $\dist_u(G_{d,l}'')\leq \dist_u(G_{d,l}')$ holds, since in that case no edge towards a vertex in the tree $T_d$ can yield a larger distance cost decrease for agent $u$.

 First, we need the value of $\dist_u(G_{d,l}')$. We have
 \begin{align*}
  \dist_u(G_{d,l}') &= 3 + 3l + \sum_{i=1}^{\frac{d}{2}+1}\dist_u(H_{i-1,i}) + \sum_{i=2}^{\frac{d}{2}}\dist_u(H_{d-i+1,i})\\
  &= 3+ 3l + \sum_{i=1}^{\frac{d}{2}+1}\left(2^i(2i-1)+1\right) + \sum_{i=2}^{\frac{d}{2}}\left(2^{d-i+2}(d+1)+1\right)\\
  &= 3l + 2^{d+1}d + d -2^{\frac{d}{2}+3} + 2^{d+1}+9.
 \end{align*}
 Thus, we have
 \begin{align*}
  &\dist_u(G_{d,l}'') \leq \dist_u(G_{d,l}')\\
  \iff &  2^{d+1}d + d - 2^{\frac{d}{2}+2} - 2^{\frac{d}{2}+3} + 2^{d+2} + 9 - \left(2^{d+1}d + d -2^{\frac{d}{2}+3} + 2^{d+1}+9\right)  \leq l\\
  \iff & 2^{d+1} - 2^{\frac{d}{2}+2} \leq l.
 \end{align*}
\end{proof}
\noindent Before we can prove the lower bound, we have to investigate for which edge-prices $\alpha$ the network $G_{d,l}$ is in \kGE. For this, we investigate which agent in $G_{d,l}$ can gain most by $k$-local greedy move and then we analyze the maximum possible cost decrease of this agent. We start by analyzing the best possible $k$-local greedy moves by agents in $V(T_d)$.
\begin{lemma}\label{lem_bestmove_treeagent}
 Let $l \geq 2^{d+1}$. For any agent $v \in V(T_d)$ an optimal single $k$-local edge-purchase reduces the distance to vertex $z$ by the maximum possible amount.
\end{lemma}
\begin{proof}
 Let $S(x)$ be the set of vertices in the subtree rooted at $x$ in $T_d$.

 First of all, we consider what happens, if agent $v$ buys an edge within her $k$-neigh\-bor\-hood towards some vertex $w \in S(v)$. We claim that the target vertex $w\in S(v)$, which yields the largest cost decrease for $v$, must have distance $2$ towards $v$, independent of the value of $k$, as long as $k\geq 2$ holds. For this, assume towards a contradiction that an edge to $w\in S(v)$, with $d_{G_{d,l}}(v,w) = j \geq 3$, yields the largest distance cost decrease for agent $v$. Let $w' \in S(v)$ be the neighbor of $w$ which lies on $w$'s shortest path to $v$ and let $w''$ be the other neighbor of $w'$ having the same distance to $v$ as vertex $w$. We claim that buying the edge $vw'$ yields a strictly larger distance cost decrease for $v$ than buying the edge $vw$. With edge $vw'$ instead of edge $vw$ agent $v$'s distances to all vertices in $S(w)$ increase by $1$ but her distances to vertices in $S(w'')$ decrease by $1$. But her distance to $w'$ decreases by one as well. Thus, buying edge $vw'$ yields strictly less cost than buying edge $vw$. Note that this argument can be applied if $d_{G_{d,l}}(v,w) \geq 3$. Thus, since $v$ already owns an edge towards her neighbors in $T_d$, the best possible purchase is buying an edge towards a vertex in distance $2$ in $S(v)$. With this, the maximum possible distance cost decrease for any vertex $v \in V(T_d)$ is $2^{d-1}-1$, which is obtained if the root $r$ buys an edge towards vertex $t \in S(r)$ having distance $2$ to $r$. This move decreases $r$'s distances to all agents in $S(t)$ by $1$ and we have $|S(t)| = 2^{d-1}-1$, since $T_d$ is a complete binary tree of depth $d$.

 A similar argument can be applied if agent $v$ buys an edge towards some vertex $w \notin S(v)$, which does not reduce $v$'s distance to $z$ by the maximum possible amount. Let $d_x = d_{G_{d,l}}(x,z)$ for any $x \in V(G_{d,l})$. There are two cases:
 \begin{enumerate}
  \item If $d_v \leq k$, then $w=z$ must hold. This is true since if $w\neq z$, there is a neighbor $w'$ of $w$ where $d_{w'}<d_w$ holds. Connecting to $w'$ instead of $w$ reduces $v$'s distances to all $l+1 \geq 2^{d+1}+1$ vertices of the star by $1$ whereas it increases $v$'s distances to at most $2^{d+1}$ vertices, that is, to all vertices in $V(T_d)\cup\{y\}$, by $1$.
  \item If $d_v > k$, then $d_w = d_v-k$ must hold. If $d_w \geq d_v-k+1$, then this implies that there is some vertex $w^*$ in $v$'s $k$-neighborhood for which $d_{w^*} = d_v-k$ holds and we have $w \neq w^*$. Thus there must be a neighbor $w'$ of $w$ which is closer to $w^*$, where $w' = w^*$ is possible. If $v$ buy the edge $vw^*$ instead of the edge $vw$, then $v$'s distances to at least $l+1 = 2^{d+1}+1$ many vertices decrease by $1$ and her distances to at most $2^{d+1}-1 = |S(r)|$ many vertices increase by $1$.
 \end{enumerate}
 If follows, that the best edge-purchase towards any vertex $w\notin S(v)$ must decrease $v$'s distance to $z$ by the maximum possible amount. Since this amount is at least $1$, because no vertex of $V(T_d)$ is a neighbor of $z$, which yields a distance cost decrease of at least $l+1 \geq 2^{d+1}+1$. Thus, any edge-purchase towards some vertex in $S(v)$ cannot be optimal.
\end{proof}
\noindent Now we analyze which agent can gain most by a $k$-local greedy move.
\begin{lemma}\label{lem_leaflargestdecrease}
 Let $G_{d,l}$ be a $d$-$l$-Tree-Star where $l\geq 2^{d+1}$. Any leaf agent $u$ of the tree $T_d$ is an agent in $V(G_{d,l})$ who can decrease her cost most by performing a $k$-local greedy move.
\end{lemma}
\begin{proof}
 $G_{d,l}$ is a tree, which implies that no agent can delete any edge and, by the special structure of the edge-ownership, it is obvious that no agent can improve by swapping any own edge. Thus, we only have to analyze edge-purchases. 

 First we consider agents in $V(T_d)$. By Lemma~\ref{lem_bestmove_treeagent}, we have that the best possible edge-purchase of such an agent must decrease this agent's distance to $z$ by the maximum possible amount. We claim that leaf agents of $T_d$ can achieve the largest distance cost decrease among all agents in $V(T_d)$. For this, assume towards a contradiction that there is some non-leaf agent $v \in V(T_d)$ who achieves the largest distance cost decrease by buying an edge towards a vertex $w$ in her $k$-neighborhood. By Lemma~\ref{lem_bestmove_treeagent}, $d_{G_{d,l}}(w,z) = d_{G_{d,l}}(v,z)-k$ must hold. Let $v'$ be a neighbor of $v$ having larger distance to the root $r$. Such vertex exists, since $v$ is not a leaf of $T_d$. Let $w'$ be a neighbor of $w$ which lies on the shortest path from $v$ to $w$. We claim that agent $v'$ can achieve a strictly larger cost decrease by buying the edge $v'w'$ than agent $v$ by buying the edge $vw$. Let $A_i \subset V(G_{d,l})$ be the set of vertices to which agent $v$ decreases her distance by exactly $i$ by buying the edge $vw$ and let $A_i' \subset V(G_{d,l})$ be the set of vertices to which agent $v'$ decreases her distance by $i$ by buying the edge $v'w'$. Thus, agent $v$ has distance cost decrease of $\sum_{i=1}^{k-1}i\cdot|A_i|$ by buying edge $vw$ and agent $v'$ has distance cost decrease of $\sum_{i=1}^{k-1}i\cdot|A_i'|$ by buying edge $v'w'$. Since $w'$ lies on the shortest path from agent $v$ to $w$ and since $w$, by Lemma~\ref{lem_bestmove_treeagent}, lies on the shortest path from $v$ to $z$, we have that $A_i \subset A_i'$ holds for all $1\leq i \leq k-1$. Thus, the claim follows.

 It remains to compare the maximum possible distance cost decrease by any agent of $V(G_{d,l})\setminus V(T_d)$ with the distance cost decrease of some leaf agent of $T_d$. Let $x \in V(G_{d,l})\setminus V(T_d)$ and let $u$ be some leaf agent of $T_d$. We show that the maximum possible distance cost decrease of agent $x$ for any $k$ is at most the distance cost decrease of agent $u$ for $k=2$. Since $u$'s maximum distance cost decrease for $k\geq 3$ must be at least her maximum distance cost decrease for $k=2$, this then finishes the proof.

 The maximum possible decrease in distance cost for any agent $x \in V(G_{d,l})\setminus V(T_d)$ and for any $k$ is achieved if some leaf agent of the star, say $l_1$, buys the edge $l_1r$. This can be seen by an analogous argument as in the first part of the proof of Lemma~\ref{lem_bestmove_treeagent}. With this move, agent $l_1$'s distance cost decreases by $2|V(T_d)| = 2^{d+2}-2$.

 In comparison, if agent $u$ buys an edge to $u^*$, where $u^*$ is on $u$'s shortest path to $z$ and $u^*$ has distance $2$ towards $u$ in $G_{d,l}$, then agent $u$'s distance to all but $3$ vertices decreases by $1$. Thus, $u$'s distance cost decrease is $|V(T_d)| + 2 + l -3 = 2^{d+1}-2 + 2^{d+1} = 2^{d+2}-2$.
\end{proof}
\noindent We are left to analyze the exact amount of agent $u$'s maximum possible distance cost decrease, depending on $d,k$ and $l$.
\begin{lemma}\label{lem_maxcostdecrease}
Let $G_{d,l}$ be a $d$-$l$-Tree-Star. The maximum possible distance cost decrease achieved by any agent for any $2 \leq k\leq d$ and $l\geq 2^{d+1}$ is
$$\Delta_{d,k,l} = (k-1)\left(l+2^{d+1}\right) -2k\left(2^{\left\lceil\frac{k}{2}\right\rceil}-1\right) + \left\lfloor\frac{k}{2}\right\rfloor 2^{\left\lceil\frac{k}{2}\right\rceil+2}-2^{k+2}+3^{\left\lceil\frac{k}{2}\right\rceil+1}-2.
 $$
\end{lemma}
\begin{proof}
By Lemma~\ref{lem_leaflargestdecrease} and Lemma~\ref{lem_bestmove_treeagent}, it suffices to compute the maximum possible distance cost decrease which any leaf agent $u$ of $T_d$ can achieve by buying one edge towards the vertex $w$ which lies on $u$'s shortest path to $z$. Moreover this edge $vw$ must decrease agent $u$'s distance to $z$ by the maximum possible amount.

Let $G_{d,l}^*$ be the network $G_{d,l}$ after agent $u$ has bought the edge $vw$.
Thus, agent $u$'s distance cost decrease is $$\Delta_{d,k,l} = \dist_u(G_{d,l}) - \dist_u(G_{d,l}^*) = \begin{cases}|V(G_{d,l})|-3, &\text{ if } k=2\\
              2(|V(G_{d,l})|-7), &\text{ if } k=3,                                                                                                                                                                       \end{cases}$$
 The general formula for any $d\geq k \geq 2$ is:
\begin{align*}
 \Delta_{d,k,l} &= \dist_u(G_{d,l}) - \dist_u(G_{d,l}^*) \\
  &= (k-1)\left(2+l+ \sum_{i=k}^d \left|V(H_{i-1,k-i+1})\right|\right) + \sum_{i=1}^{\left\lfloor\frac{k}{2}\right\rfloor-1}(k-2i-1)|V(H_{k-i-1,i})|\\
  &= (k-1)\left(2 + l + \sum_{i=k}^d 2^i \right) + \sum_{i=1}^{\left\lfloor\frac{k}{2}\right\rfloor-1}(k-2i-1)2^{k-i}\\
  &= (k-1)\left(2 + l + 2^{d+1}-2^k\right)+ \sum_{i=1}^{\left\lfloor\frac{k}{2}\right\rfloor-1}(k-2i-1)2^{k-i}\\
  &= (k-1)\left(l+2^{d+1}\right) -2k\left(2^{\left\lceil\frac{k}{2}\right\rceil}-1\right) + \left\lfloor\frac{k}{2}\right\rfloor 2^{\left\lceil\frac{k}{2}\right\rceil+2}-2^{k+2}+3^{\left\lceil\frac{k}{2}\right\rceil+1}-2.
\end{align*}
Where we have used that $|V(H_{j,q})| = 2^{j+1}$ for any $q$. Note that for $k=2$ and $k=3$ the right sum is $0$.
\end{proof}
\noindent Lemma~\ref{lem_maxcostdecrease} implies the following statement.
\begin{corollary}\label{cor_treestar_k_greedy_stable}
 Let $l \geq 2^{d+1}$, then network $G_{d,l}$ is in \kGE for any $2 \leq k \leq d$, if $\alpha \geq \Delta_{d,k,l}$.
\end{corollary}
\noindent Now we can finally set out for proving Theorem~\ref{thm_kGe_vs_Ge_approx_lower_bound}.
 We will use the network $G_{d,l}$ for constructing the lower bound for the approximation. For this, we will assume that $d$ is even and that $l \geq 2^{d+1}$ holds. Let $u$ and $z$ be defined as above and let $G_{d,l}''$ be the network obtained from $G_{d,l}$ if agent $u$ buys the edge $uz$. By Lemma~\ref{lem_lowerbound_for_l}, we have that this edge-purchase is the best possible single edge agent $u$ can buy in network $G_{d,l}$. By Corollary~\ref{cor_treestar_k_greedy_stable}, we have that $G_{d,l}$ is in \kGE if $\alpha \geq \Delta_{d,k,l}$.
 The corresponding ratio of agent $u$'s cost before and after the purchase of edge $uz$ is
 $$\frac{\dist_u(G_{d,l})}{\alpha + \dist_u(G_{d,l}'')} \leq \frac{\dist_u(G_{d,l})}{\Delta_{d,k,l} + \dist_u(G_{d,l}'')}. $$
 By equality~(\ref{eq:cost_leaf}) we have $\dist_u(G_{d,l}) = 2^{d+1}(2d-3)+3d+l(d+3)+9.$
 Lemma~\ref{lem_maxcostdecrease} yields
 \begin{align*}
  \Delta_{d,k,l} &= (k-1)\left(l+2^{d+1}\right) -2k\left(2^{\left\lceil\frac{k}{2}\right\rceil}-1\right) + \left\lfloor\frac{k}{2}\right\rfloor 2^{\left\lceil\frac{k}{2}\right\rceil+2}-2^{k+2}+3^{\left\lceil\frac{k}{2}\right\rceil+1}-2.
 \end{align*}
 Since $d$ is even, Lemma~\ref{lem_newleafcost} yields:
 \begin{align*}
  \dist_u(G_{d,l}'') &= 2^{d+1}d + d -3\cdot 2^{\frac{d}{2}+2} + 2^{d+2} + 2l + 9.
 \end{align*}
First, we consider what happens if $l$ tends to infinity:
\begin{align*}
 \lim_{l\to\infty}\frac{\dist_u(G_{d,l})}{\Delta_{d,k,l} + \dist_u(G_{d,l}'')} &= \lim_{l\to \infty}\frac{l(d+3)}{(k+1)l} = \frac{d+3}{k+1}
\end{align*}
This shows, that the approximation ratio for any constant neighborhood size $k$ may exceed any constant ratio $c$ by choosing $d>c(k+1)-3$ and $l$ large enough.

It suffices if $l$ grows fast enough to dominate the other terms in the numerator and in the denominator as $l$ tends to infinity.
Thus, we choose $l = 3^d \iff d = \log_3 l$, which yields
\begin{align*}
 \lim_{l\to\infty}\frac{\dist_u(G_{\log_3 l,l})}{\Delta_{\log_3 l,k,l} + \dist_u(G_{\log_3 l,l}'')} &= \frac{\Omega\left(l \log_3 l\right)}{\mathcal{O}(k\cdot l)} = \Omega\left(\frac{\log_3 l}{k}\right) = \Omega\left(\frac{\log l}{k}\right).
 \end{align*}
Since $n = 2^{d+1}+1+l$, and since $2^{d+1}+1 \leq 3^d$ for $d\geq 2$, we have that $n \geq l \geq \frac{n}{2}$, which implies that $\log l \in \Theta(\log n)$ and a lower bound of $\Omega\left(\frac{\log_3 \frac{n}{2}}{k}\right) = \Omega\left(\frac{\log n}{k}\right)$ for $2\leq k\leq \log_3 \frac{n}{2}$ and $\Omega\left(\log n\right)$ for any constant $k\geq 2$. Note that the diameter of $G_{d,l}$ is $d+3 \in \Theta(\log l)$. This finishes the proof of Theorem~\ref{thm_kGe_vs_Ge_approx_lower_bound}.
\end{proof}

\begin{proof}[Proof of Lemma~\ref{lem_tree_swap}]
 Suppose that agent $u$ can strictly decrease her cost by performing the edge-swap $uv\to uw$ in the tree network $T$. Clearly, if $\dist_T(u,w) \leq k$ then we are done. Hence, we assume that $\dist_T(u,w) = l > k \geq 2$ holds.

 First of all, observe that $v$ must be the first vertex on $u$'s shortest path to $w$ in $T$, since otherwise the removal of edge $uv$ would disconnect $T$. Let $T_v$ be subtree of vertex $v$ in the tree $T$ rooted at vertex $u$. It follows that only distances towards vertices in $V(T_v)$ may change.

 We claim that if the swap $uv\to uw$ is a best possible single edge-swap for agent $u$, then agent $u$ has an improving $2$-local swap in $T$.

 Let $P_{uw}$ be the path from $u$ to $w$ in $T$, thus we have that $$P_{uw} = u,v_1,v_2,\dots,v_{l-1},v_{l} $$ where $v_1 = v$ and $v_{l} = w$. For all vertices $z$ on $P_{uw}$ let $V_z$ denote the set of vertices of $T$ which have node $z$ on their shortest path to any neighbor of $z$ on the path $P_{uw}$. Let $T^i$ be the tree $T$ after $u$ has performed the edge-swap $uv \to uv_i$ in $T$. Thus, we have that $T^1 = T$ and $T^l$ is the tree after $u$ has performed her best possible single edge-swap $uv\to uw$. Since the swap $uv\to uw$ is a best possible edge-swap, we have that \begin{equation}\cost_u\left(T^{i}\right) \geq \cost_u\left(T^l\right), \label{ineq_cost}\end{equation} for $2\leq i \leq l-1$ must hold. Note that we have at least inequality (\ref{ineq_cost}) for $i=2$, since $l \geq 3$. Furthermore we have $$\cost_u\left(T^l\right) = \sum_{i=1}^l (l-i+1)|V_{v_i}| + \sum_{z \in V(T)\setminus V(T_v)} d_T(u,z) + \edge_u(T)$$ and $$\cost_u\left(T^{l-1}\right) = \sum_{i=1}^{l-1}(l-i)|V_{v_i}| + 2|V_{v_l}|+ \sum_{z \in V(T)\setminus V(T_v)} d_T(u,z) + \edge_u(T)$$ which, with inequality (\ref{ineq_cost}), yields
 \begin{align*}
  & \cost_u\left(T^{l-1}\right) - \cost_u\left(T^l\right) \geq 0\\
  \iff & \sum_{i=1}^{l-1}(l-i)|V_{v_i}| + 2|V_{v_l}| - \left(\sum_{i=1}^l (l-i+1)|V_{v_i}|\right) \geq 0\\
  \iff & - \sum_{i=1}^{l-1}|V_{v_i}| + |V_{v_l}| \geq 0\\
  \iff & |V_{v_l}| - \sum_{i=2}^{l-1}|V_{v_i}| \geq |V_{v_1}|\\
  \Rightarrow & |V_{v_l}| + \sum_{i=2}^{l-1}|V_{v_i}| > |V_{v_1}|\\
  \iff & \sum_{i=2}^l |V_{v_i}| > |V_{v_1}|.
 \end{align*}
Thus, it follows that the $2$-local edge-swap $uv\to uv_2$ is improving, since this swap decreases $u$'s distance to exactly $\sum_{i=2}^l|V_{v_i}|$ many agents by one each and it increases $u$'s distance to $|V_{v_1}|$ many agents by one each.
\end{proof}

\begin{proof}[Proof of Theorem~\ref{thm_tree_approx_upperbound}]
 We will show that any tree network $T$ in \kGE is in $\mathcal{O}\left(\frac{\diam(T)}{k}\right)$-approx\-imate Greedy Equilibrium, where $\diam(T)$ is the diameter of $T$. The theorem follows from this, since the contra-positive statement of Lemma~\ref{lem_tree_swap} guarantees that any tree in \kGE tree must also be in Asymmetric Swap Equilibrium and it is already known~\cite{MS12} that any $n$-vertex tree in Asymmetric Swap Equilibrium has diameter $\mathcal{O}(\log n)$.

 Let $v_0$ be any agent in $T$ who can buy an edge to decrease her cost. We assume that $v_0$ buys the edge $v_0v_l$ and that this is the best possible single edge-purchase for agent $v_0$ in $T$. Let $$P_{v_0v_l} = v_0,v_1,v_2,\dots,v_{k-1},v_k,v_{k+1},\dots,v_{l-1},v_l$$ be $v_0$'s path to $v_l$ in $T$.  Let $T_z$ be the subtree of vertex $z$ in the tree $T$ rooted at $v_0$ and let the sets $V_z$ for all $z \in V(P_{v_0v_l})$ be defined as in the proof of Lemma~\ref{lem_tree_swap}. Fig.~\ref{fig:upper_bound_notation} gives an illustration of these sets.

 \begin{figure}[h!]
 \centering
 \includegraphics[width=\textwidth]{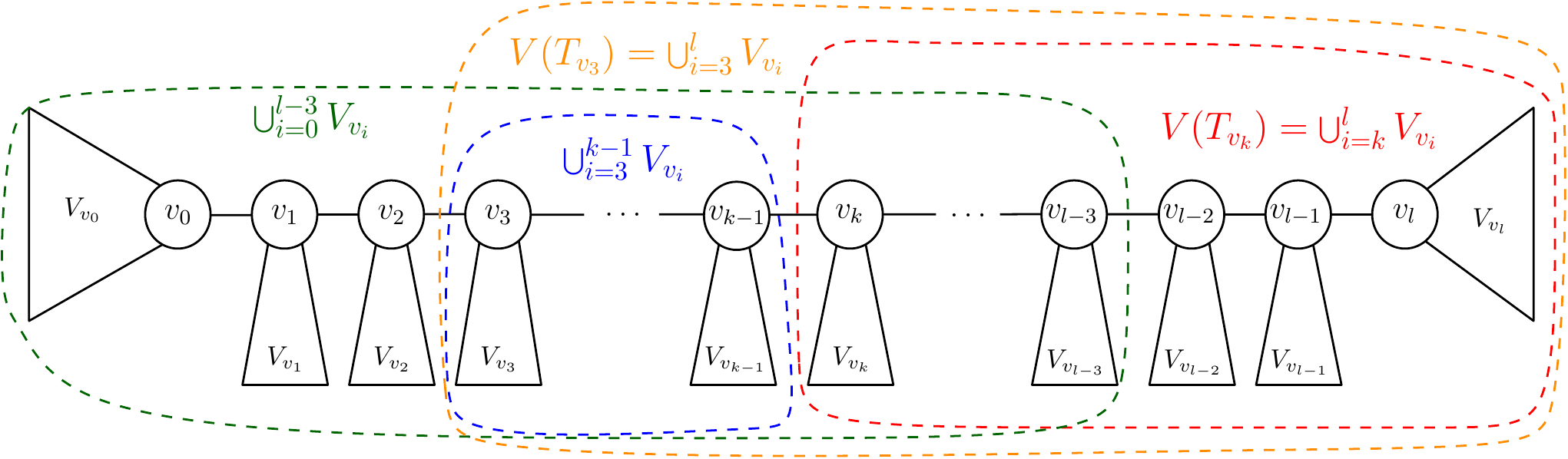}
 \caption{Illustration of the sets of vertices used in the proof.}
 \label{fig:upper_bound_notation}
\end{figure}

 We assume that agent $v_0$ cannot decrease her cost by buying an edge towards any vertex in her $k$-neighborhood. Hence we have that $\dist_T(v_0,v_l) = l > k \geq 2$ must hold.
 Since we assume that $T$ is $k$-greedy stable it follows that agent $v_0$ cannot decrease her cost by buying an edge towards $v_k$. Since this move would at least decrease $v_0$'s distance to all vertices in $V\left(T_{v_k}\right)$ by $k-1$, we have \begin{equation}
           \alpha \geq (k-1)|V(T_{v_k})|. \label{eq_alpha_bound}                                                                                                                                                                                                                                                                                                                                                                                                                                                                                                                                                                                                                                                                                                                                                                                                                                                                                                                                                                                                                                                                                                                                                                                                                                                                                                                                                                                                                                                                                                                                                                                                                                                                                            \end{equation}
 Now we consider the ratio of agent $v_0$'s cost before and after the purchase of the edge $v_0v_l$. Let $T'$ be the tree $T$ after $v_0$ has bought the edge $v_0v_l$. Moreover, let $\delta_{v_0}$ denote the distance cost decrease for agent $v_0$ obtained by buying the edge $v_0v_l$ in $T$. Thus, the ratio is:
 \begin{align*} \frac{\cost_{v_0}(T)}{\cost_{v_0}(T')}  &= \frac{\cost_{v_0}(T)}{\cost_{v_0}(T)-\delta_{v_0} + \alpha}
  = \frac{\edge_{v_0}(T) + \sum_{v\in V(T)}d_T(v_0,v)}{\edge_{v_0}(T) + \sum_{v\in V(T)}d_T(v_0,v) - \delta_{v_0} + \alpha}\\
  &\leq \frac{\sum_{v \in V(T_{v_3})}d_T(v_0,v)}{\sum_{v \in V(T_{v_3})}d_T(v_0,v)- \delta_{v_0} + \alpha}.
 \end{align*}
 The last inequality holds, since all vertices to which $v_0$ decreases her distance by buying the edge $v_0v_l$ must lie in the subtree $T_{v_3}$.
 
 We can upper bound the nominator easily by assuming that all vertices in $V(T_{v_3})$ are in maximum distance to $v_0$, that is, we have
 $$ \sum_{v \in V(T_{v_3})}d_T(v_0,v) \leq \diam(T)|V(T_{v_3})|.$$
 Since $v_0$ has at least distance $1$ to all vertices in $V(T_{v_3})$ after buying the edge $v_0v_l$, we have that $\sum_{v \in V(T_{v_3})}d_T(v_0,v)- \delta_{v_0} > 0$ must hold. Thus, we can lower bound the denominator as follows: $$\sum_{v \in V(T_{v_3})}d_T(v_0,v)- \delta_{v_0} + \alpha > \alpha.$$
 Hence, we have
 $$\frac{\cost_{v_0}(T)}{\cost_{v_0}(T')} \leq \frac{\sum_{v \in V(T_{v_3})}d_T(v_0,v)}{\sum_{v \in V(T_{v_3})}d_T(v_0,v)- \delta_{v_0} + \alpha} \leq \frac{\diam(T)|V(T_{v_3})|}{\alpha}. $$
 Using inequality (\ref{eq_alpha_bound}), we get
 $$\frac{\cost_{v_0}(T)}{\cost_{v_0}(T')} \leq \frac{\diam(T)|V(T_{v_3})|}{\alpha} \leq \frac{\diam(T)|V(T_{v_3})|}{(k-1)|V(T_{v_k})|}.  $$
 Note, that for $k\leq 3$ this already yields $$\frac{\cost_{v_0}(T)}{\cost_{v_0}(T')} \in \mathcal{O}(\diam(T))$$ for $k\leq 3$, since $|V(T_{v_i})| \geq |V(T_{v_3})|$ for $i\leq 3$.
 Towards a general upper bound of $\mathcal{O}\left(\frac{\diam(T)}{k}\right)$, it remains to show that $|V(T_{v_k})| \in \Omega(|V(T_{v_3})|)$ holds for $k>3$.

 So far, we have not yet used the assumption, that buying the edge $v_0v_l$ is the best possible edge-purchase for agent $v_0$. From this assumption follows that buying the edge $v_0v_{l-1}$ instead of $v_0v_l$ cannot be more profitable for agent $v_0$. Swapping from $v_0v_l$ to $v_0v_{l-1}$ would increase agent $v_0$'s distances to all vertices in $V_{v_l}$ by one each and it would decrease agent $v_0$'s distances to all vertices in the set $\bigcup_{i=\left\lfloor\frac{l}{2}\right\rfloor+1}^{l-1} V_{v_i}$ by one each. Since all sets $V_{v_i}$ are pairwise disjoint by definition, it follows that \begin{equation} \sum_{i = \left\lfloor\frac{l}{2}\right\rfloor+1}^{l-1}|V_{v_i}| \leq |V_{v_l}|. \label{ineq_tree_numbers} \end{equation}
 If $k = 4 = l-1$, that is, $l-2 = 3$, then, since buying $v_0v_l$ strictly decreases agent $v_0$'s cost whereas buying edge $v_0v_k =v_0v_{l-1}$ does not, we have that
 \begin{equation}|V_{v_l}|>\sum_{i=\left\lfloor\frac{l}{2}\right\rfloor+1}^{l-1}|V_{v_i}| \label{ineq_k_l-1}\end{equation} holds. Thus, we have $\left\lfloor\frac{l}{2}\right\rfloor+1 = 3$ and thus we have that $|V(T_{v_3})|<2|V_{v_l}|$ which implies that $|V(T_{v_3})|<2|V(T_{v_k})|$, which yields $|V(T_{v_k})| \in \Omega(V(T_{v_3}))$.

 If $l-2 > 3$, then we claim that the edge $v_{l-2}v_{l-1}$ must be owned by agent $v_{l-1}$. This is true, since otherwise agent $v_{l-2}$ could perform the swap $v_{l-2}v_{l-1} \to v_{l-2}v_l$ and thereby strictly decrease her cost. This can be seen as follows: If $l-1 = k$, then, by inequality~(\ref{ineq_k_l-1}), we have that $|V_{v_l}| > |V_{v_{l-1}}| = |V_{v_k}|$ holds. On the other hand, if $l-1>k$, then inequality~(\ref{ineq_tree_numbers}) directly implies $|V_{v_{l-1}}|<|V_{v_l}|$, since the sum on the left as at least one additional non-zero summand. In both cases we have that the swap $v_{l-2}v_{l-1} \to v_{l-2}v_l$ must be improving for agent $v_{l-2}$. This proves the claim.

 Having established that the edge $v_{l-2}v_{l-1}$ is owned by agent $v_{l-1}$ and using the assumption that no agent in $T$ can swap an edge in her $k$-neighborhood to strictly decrease her cost, it follows that the swap $v_{l-1}v_{l-2}\to v_{l-1}v_{l-3}$ cannot be improving for agent $v_{l-1}$, which yields
  $|V_{l-2}| \geq \sum_{i=0}^{l-3}|V_{v_i}|. $
  Since $l>5$, we have that $\left\lfloor\frac{l}{2}\right\rfloor+1 \leq l-2$ which, by inequality~(\ref{ineq_tree_numbers}), implies $ \sum_{i=0}^{l-3}|V_{v_i}| \leq |V_{v_{l-2}}| < \sum_{i = \left\lfloor\frac{l}{2}\right\rfloor+1}^{l-1}|V_{v_i}| \leq |V_{v_l}|.$
  Thus, if $k = l-1$ we have that $\sum_{i=3}^{k-1}|V_{v_i}| \leq 2|V_{v_{l-2}}| \leq 2|V_{v_l}| \leq 2|V(T_{v_k})|,$ which implies that $|V(T_{v_3})| \leq 3|V(T_{v_k})|$.
  If $k \leq l-2$, it follows that $\sum_{i=3}^{k-1}|V_{v_i}| < |V(T_{v_k})|$ which implies $|V(T_{v_3})| < 2|V(T_{v_k})|$. 
  
  In both cases this yields $|V(T_{v_k})| \in \Omega(|V(T_{v_3})|).$
\end{proof}

\begin{proof}[Proof of Lemma~\ref{lem_non-tree_swap}]
 We use the non-tree network $G_k$ depicted in Fig.~\ref{fig:kswap_vs_k+1swap} and we focus on agent $u$.
 \begin{figure}[h!]
 \centering
 \includegraphics[width=0.8\textwidth]{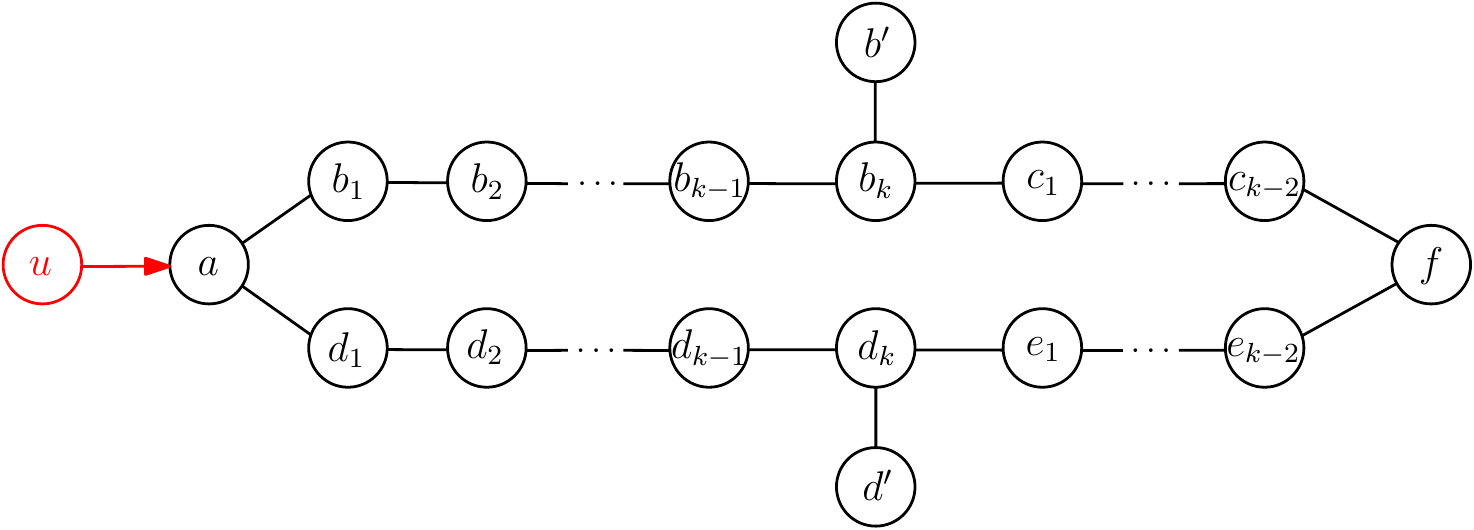}
 \caption{A non-tree network where agent $u$ cannot improve by a $k$-local swap but by a $k+1$-local swap.}
 \label{fig:kswap_vs_k+1swap}
\end{figure}

First, we show that agent $u$ cannot improve by performing a single edge-swap in her $k$-neighborhood. By symmetry, it suffices to consider the edge-swaps $ua \to ub_i$ for $1\leq i \leq k-1$. We prove that all those swaps are cost neutral, that is, that they are not improving for agent $u$.
Clearly, since the set $V(G_k) \setminus \{u,b',d'\}$ induces a simple cycle, we have that agent $u$'s total distance cost towards all vertices in $V(G_k) \setminus \{u,b',d'\}$ must be the same after the swap $ua \to ub_{i}$. The only possibility for a change in $u$'s distance cost is, if the swap $ua \to ub_i$ yields a larger distance decrease towards vertex $b'$ than the incurred distance increase towards vertex $d'$ (or vice versa). The swap $ua \to ub_i$ decreases $u$'s distance to $b'$ by exactly $i$ from $k+2$ to $k+2-i$. Let $G_k'$ be the network after $u$ has performed the swap $ua\to ub_i$. In network $G_k'$ agent $u$ has two possible paths towards vertex $d'$. The path via vertex $a$ has length $$d_{G_k'}(u,b_i)+d_{G_k'}(b_i,a)+d_{G_k'}(a,d') = 1 + i + k+1 = k+2+i.$$ The path via vertex $f$ has length $$d_{G_k'}(u,b_i) + d_{G_k'}(b_i,b_k) + d_{G_k'}(b_k,f) + d_{G_k'}(f,d') = 1 + k-i + k-1 + k = 3k -i$$ and we have $3k-i \geq k+2-i$. Hence, the path via $a$ is $u$'s shortest path to $d'$ and $u$'s distance to $d'$ increases by exactly $i$ from $k+2$ to $k+2+i$. Thus, the distance increase to $d'$ exactly neutralizes her distance decrease to $b'$. Thus, all swaps $ua \to ub_i$ for $1\leq i \leq k-1$ are cost neutral for agent $u$.

It remains to show that agent $u$ can perform a $k+1$-local edge-swap to strictly decrease her cost. With the argumentation above, we only have to analyze $u$'s distance change to $b'$ and $d'$. We consider the swap $ua\to ub_k$, which yields a distance decrease of $k+1$ towards vertex $b'$. Let $G_k''$ be the network after $u$'s swap $ua \to ub_k$. In $G_k''$ agent $u$'s distance to vertex $d'$ is
$$d_{G_k''}(u,d') \leq d_{G_k''}(u,b_k) + d_{G_k''}(b_k,f) + d_{G_k''}(f,d') = 1 + k-1 + k = 2k. $$ Thus, agent $u$'s distance increase to $d'$ is at most $k-2$ which is strictly less than her distance decrease to $b'$. It follows that the $k+1$-local swap $ua\to ub_k$ is improving for agent $u$.
\end{proof}

\begin{proof}[Proof of Theorem~\ref{thm_approx_upperbound}]
 First, we consider single edge-swaps. We assume that agent $u$ performs an edge-swap in network $G$ and the obtained network is $G'$. Let $D$ be the set of vertices to which the edge-swap strictly decreases agent $u$'s distance. And let $I$ be the set of vertices to which the edge-swap strictly increases agent $u$'s distance. We assume that the edge-swap is an improving move for agent $u$. We have that agent $u$'s distance-cost decrease $\delta_u$ is $$ \delta_u = \sum_{x\in D}(d_G(u,x) - d_{G'}(u,x)) - \sum_{x\in I}(d_{G'}(u,x) - d_G(u,x)).$$
 Clearly, we have $$\delta_u < \sum_{x\in D}(d_G(u,x) - d_{G'}(u,x)).$$ Thus, it follows that
 \begin{align*} 
   \frac{\cost_u(G)}{\cost_u(G')} &= \frac{\cost_u(G)}{\cost_u(G)-\delta_u} = \frac{\edge_u(G) + \dist_u(G)}{\edge_u(G') + \dist_u(G') - \delta_u} < \frac{\dist_u(G)}{\dist_u(G') - \delta_u}\\
   &= \frac{\sum_{x\in D}d_G(u,x) + \sum_{x \in V\setminus D}d_G(u,x)}{\sum_{x\in D}d_G(u,x) + \sum_{x \in V\setminus D}d_G(u,x) - \delta_u}
   \leq \frac{\sum_{x\in D}d_G(u,x)}{\sum_{x\in D}d_G(u,x) - \delta_u}\\
   &< \frac{\sum_{x\in D}d_G(u,x)}{\sum_{x\in D}d_G(u,x) -\left(\sum_{x\in D}(d_G(u,x) - d_{G'}(u,x))\right)}\\
   &\leq \frac{\diam(G)|D|}{\sum_{x\in D}d_{G'}(u,x)} \leq \frac{\diam(G)|D|}{|D|} = \diam(G).
 \end{align*}
 For single edge-purchases we have the following:
  \begin{align*}
  \frac{\cost_u(G)}{\cost_u(G')} &= \frac{\cost_u(G)}{\cost_u(G)-\delta_u + \alpha}
  \leq \frac{\sum_{x\in D}d_G(u,x)}{\sum_{x\in D}d_G(u,x) - \delta_u + \alpha}\\ &\leq \frac{\sum_{x\in D}d_G(u,x)}{|D|+\alpha} < \frac{\diam(G)|D|}{|D|} = \diam(G).
 \end{align*}
 The second inequality holds since in $G'$ agent $u$ must have at least distance $1$ towards all vertices in $D$.
\end{proof}

\section{Omitted Details from Section~\ref{section:quality-of-equilibria}}
\label{section:appendix_quality-proofs}

\begin{proof}[Proof of Theorem~\ref{thm_approx_diam_connection}]
Let $u$ be an arbitrary agent of $G$ and consider $T$ being a shortest path tree rooted at $u$.
(Note that paths to all other agents exist, since $G$ is in $\beta$-approximate \NE.)
For every agent $v\in V$, we consider the strategy change of removing all own edges that do not belong to $T$ and creating one new edge to $u$.
Let $T_v\subseteq T$ be the set of tree edges owned by $v$.
Since $G$ is in $\beta$-approximate \NE and $v$ is not changing $\dist_u(G)$ by this move, we get
$
    \cost_v(G,\alpha) \leq \beta(|T_v|\alpha + \alpha + (n-1) + \dist_u(G))
    .
$
Hence, for the social cost we get:
\begin{align*}
    \cost(G,\alpha) &\leq \sum_{v\in V} c_v(S)
        \leq \sum_{v\in V} \beta(|T_v|\alpha + \alpha + (n-1) + \delta_u)\\
        & 
        \leq \beta(2(n-1)\alpha + (n-1)^2 + n(n-1)D)
\end{align*}
Since the optimal solution (i.e., a star) has cost of $\alpha(n-1) + n(n-1)$ we get as upper bound for the social cost ratio
$
    \beta(2 + 1 + D)
    .
$
\end{proof}

\begin{proof}[Proof of Theorem~\ref{thm:greedy_diameter}]
If $k\geq 2\sqrt{\alpha}$, then \kNE and \NE coincide, since it was shown in~\cite{Fab03} that no two agents can have distance $2\sqrt{\alpha}$, since otherwise one of them could buy an edge to the other and strictly decrease her cost. Since $k$ is large enough, this would still be a $k$-local move. Hence, for $k\geq 2\sqrt{\alpha}$ we get  $\diam(G) \leq 2\sqrt{\alpha}$.

Otherwise, if $k < 2\sqrt{\alpha}$, let $u$ be an agent with maximal eccentricity in $G$ and let $v$ be a most distant agent to $u$. Let $D\coloneqq \diam(G)$ and consider the distance cost improvement of $u$ by creating an edge to the agent $x$ at distance $k$ on the shortest path from $u$ to $v$.
When creating this edge, $u$ reduces her distance cost by $k-1$ to each of the $D - k$ last agents on the path, and in total
by $\sum_{i=1}^{\lfloor k/2\rfloor} (k-2i+1)$ to the $\lfloor\tfrac{k}{2}\rfloor$ last agents on the same path from $u$ to $x$, including $x$. Since $G$ forms a \kNE, we have:
$
    \alpha \geq (k-1)(D-k) + \sum_{i=1}^{\lfloor k/2\rfloor} (k-2i+1)
$,
which yields $D\leq \alpha/(k-1) + \tfrac{3k}{2} + 1$.
\end{proof}

\begin{proof}[Proof of Lemma~\ref{lemma:ballSizeIncrease}]
The claim directly holds, if there is a $u\in V$ with $|N_{2d+3}(u)| > n/2$.
Hence, we assume the contrary and fix an arbitrary $u \in V$.
Denote $u$'s $(2d+3)$-neighborhood as $B\coloneqq N_{2d+3}(u)$ and name the agents at distance exactly $(2d+3)$ as $\partial B\coloneqq \{v\in V| d(u,v) = 2d + 3\}$.
We now greedily select a maximal subset $X\subseteq \partial B$ by the following iterative algorithm:
(1) mark all agents of $\partial B$ as unassigned, (2) while there is an unassigned agent $x$ in $\partial B$, add $x$ to $X$ and create a new set $\partial C_x$ containing $x$ and all unassigned agents of $\partial B$ within distance of at most $2d$ to $x$, and mark these agents as assigned.
Note that for the so computed set $X$ it holds that for any two $x,y\in X$ with $x\not= y$ we have $d(x,y) > 2d$.

Next, we lower bound the number of clusters by $|X| \geq n/\alpha$.
For this, enumerate the elements of $X$ with $x_1,\ldots, x_{|X|}$ and define clusters $C_{x_i}$ such that every $C_{x_i}$ contains all elements of the corresponding $\partial C_{x_i}$ and further, for every agent $v\in V\setminus B$, we select an arbitrary shortest path from $u$ to $v$ and assign $v$ to the cluster $C_{x_i}$ that contains the (unique) agent on the path belonging to $\partial B$.
By construction, we have $|\bigcup_{i=1}^{|X|} C_{x_i}| \geq n/2$.
Now assume that $u$ buys an edge to some $x \in X$, say to $x_i$.
After this operation, $u$'s distance to every $v\in \partial C_{x_i}$ is at most $2d + 1$ and thus the distance to any $w\in C_{x_i}$ decreases by at least $2$.
Since $G$ forms an equilibrium and creating an edge to $x_i$ is $k$-local, we get $\alpha \geq 2|C_{x_i}|$ for every $x_i\in X$.
Hence, $|X|\alpha \geq 2\sum_{i=1}^{|X|} |C_i| \geq 2 n/2 = n$, i.e., $|X|\geq n/\alpha$.

By construction, for any $x,y\in X$ with $x\not=y$ we have $N_d(x)\cap N_d(y) = \emptyset$.
With $|N_d(x)| > \lambda$ this gives $|\bigcup_{x\in X} N_d(x)| > |X| \lambda$.
For every $x\in X$ we have $d(u,x) = 2d + 3$ and hence, the maximal distance from $u$ to any $v\in N_d(x)$ is at most $3d + 3$.
This gives, $|N_{3d+3}(u)| \geq |\bigcup_{x\in X} N_d(x)| > |X| \lambda \geq \lambda n/\alpha$.
\end{proof}

\begin{proof}[Proof of Lemma~\ref{lemma:halfNodesToAllBall}]
We prove the contra-positive:
Assume $|N_{2d+1}|< n$, then there is a $v\in V$ such that $d(u,v) = 2d + 2$.
Since for all $x\in N_d(u)$ it holds $d(u,x)\leq d$, by the triangle inequality we get that $d(v,x) \geq d + 2$ for all $x\in N_d(u)$.
Now consider $v$ buying an edge to $u$, which reduces $\dist_v(G)$ by at least $|N_d(u)|$.
Since $G$ forms an equilibrium, we get $n/2 > \alpha \geq |N_d(u)|$, which gives the claim.
\end{proof}

\section{Omitted Details from Section~\ref{subsection:eq-conformity}}
\label{section:kne-and-ne-conformity}
As stated in Observation~\ref{observation:NEinclusions}, for every $k\geq 1$ it holds $NE\subseteq \kNE$.
In this section, we discuss the combinations of $k$ and $\alpha$ for which the equilibria sets actually match, i.e., for which $NE = \kNE$ holds.
Note that having $NE = \kNE$ and further a $k$, which is bigger than any equilibrium diameter for a specific $\alpha$-range, we directly get the price for anarchy results that are known for Nash equilibria for this $\alpha$-range.

\begin{lemma}
For $0< \alpha < 1$ and $2\leq k$, it holds $\kNE = \NE$ and the price of anarchy is~$1$.
\end{lemma}
\begin{proof}
Given a graph $G$ and $\alpha < 1$, assume there are two (closest) agents $u,v$ that are not connected by one edge, i.e., $d_G(u,v)=2$.
In this case, buying an edge $uv$ is an improving response for $u$.
Hence, the only equilibrium graph for $\alpha < 1$ is a clique, which is also the optimal solution (cf.\ \cite{Fab03}).
\end{proof}

\begin{lemma}[\cite{De07}, Theorem~4]
For $1\leq \alpha \leq \sqrt{n/2}$ and $6\leq k$, it holds $\kNE = \NE$ and the price of anarchy is at most $6$.
\end{lemma}
\begin{proof}
In \cite{De07}, the authors show that every shortest path tree rooted at a agent $v$ has a height of at most $5$.
For this, they assume the contrary and show the existence of an improving response where a agent at distance of at least $6$ buys an edge towards $v$.
This operation is allowed with $k\geq 6$, hence there is no $k$-local equilibrium with diameter of more than $5$, i.e., $\kNE = \NE$ and the PoA bound of \cite{De07} applies.
\end{proof}

\begin{lemma}[\cite{De07}, Theorem~10]
\label{lemma_poa_upper_bound_for_smaller_n}
For $1 \leq \alpha < n^{1-\varepsilon}$, $\varepsilon \geq 1/\lg(n)$ and $k \geq 4.667\cdot 3^{\lceil 1/\varepsilon\rceil} + 8$, it holds $\kNE = \NE$ and the price of anarchy is at most $4.667\cdot 3^{\lceil 1/\varepsilon\rceil} + 8$.
\end{lemma}
\begin{proof}
In Theorem~10 of \cite{De07}, the authors use an inductive argument to find a agent $u$ and a radius $d$ such that the $d$-neighborhood of $u$ contains more than $n/2$ many agents.
For this, they start with their Lemma~3 (for which only $k\geq 2$ must hold) and  apply Lemma~9 iteratively.
They show that the maximal radius $d$, for which Lemma~9 must be applied, is at most $4.667\cdot 3^{\lceil 1/\varepsilon\rceil} + 8$, which gives a first lower bound for $k$.
Using this result, they apply their Corollary~7 to show that actually all agents are contained in a ball of radius $4.667\cdot 3^{\lceil 1/\varepsilon\rceil} + 7$, for which they need the operation of creating an edge to a agent at distance $4.667\cdot 3^{\lceil 1/\varepsilon\rceil} + 8$, which is the second lower bound for $k$.

Using both results, they show that the diameter of every equilibrium is at most $4.667\cdot 3^{\lceil 1/\varepsilon\rceil} + 8$.
By the choice of $k$, the same holds for $k$-local equilibria, i.e., $\kNE = \NE$ holds and we get a price of anarchy of at most $4.667\cdot 3^{\lceil 1/\varepsilon\rceil} + 8$.
\end{proof}

\begin{lemma}[\cite{De07}, Theorem~12]
For $1 \leq \alpha \leq 12n\log n$ and $k\geq 2\cdot 5^{1+\sqrt{\lg n}} + 24\lg(n) + 3$, it holds $\kNE = \NE$ and the price of anarchy is $\mathcal{O}\left(5^{\sqrt{\lg n}}\lg n\right)$.
\end{lemma}
\begin{proof}
Similar to their proof of Theorem~10 in \cite{De07}, the authors provide a price of anarchy upper bound for a larger range of $\alpha$.
Again, they use an inductive argument to find a agent $u$ and a radius $d$ such that the $d$-neighborhood of $u$ contains more than $n/2$ many agents.
For this, they start with looking at the number of agents in any radius $12\lg n$ neighborhood and then apply Lemma~11 iteratively.
They show that the maximal radius $d$, for which Lemma~11 must be applied, is at most $5^{1+\sqrt{\lg n}}$, which gives a first lower bound for $k$.
Using this result, they apply their Corollary~8 to show that actually all agents are contained in a specific ball, for which they need the operation of creating an edge to a agent at distance $2\cdot 5^{1+\sqrt{\lg n}} + 24\lg(n) + 3$, which is the second lower bound for $k$.

Using both, they show that in every equilibrium there is a agent that contains all other agents withing its $8\cdot 5^{1+\sqrt{\lg n}} + 24\lg(n) + 2$ neighborhood.
With the choice of $k$, the same holds for $k$-local equilibria and we get $\kNE = \NE$ as well as a price of anarchy upper bound of $\mathcal{O}\left(5^{\sqrt{\lg n}}\lg n\right)$.
\end{proof}

\begin{lemma}[\cite{Al06}, Theorem~3.1]
For $12n\log n \leq \alpha$ and $2 \leq k$, it holds $\kNE = \NE$ and the price of anarchy is $\mathcal{O}(1)$.
\end{lemma}
\begin{proof}
In \cite{Al06}, the authors provide a proof that characterizes equilibria for $12n\log n \leq \alpha$.
The main insight for their bound is that there are different types of agents that lead (see their Proposition~1, which uses their Lemma~3.1 and Lemma~3.2) to the equilibrium characterization that any equilibrium graph with girth of at least $12\lceil\log n\rceil$ has a diameter of less than $6\lceil\log(n)\rceil$ and is a tree.
In their Lemma~3.3, they prove that this big $\alpha$ ensures a girth of at least $12\lceil\log n\rceil$.
The result of Theorem~3.1 then comes from a comparison to the social optimum and gives a price of anarchy upper bound of at most $1.5$.

Interestingly, in all used statements, there are only two statements concerning creation or deletion of edges.
For Lemma~3.2, the operation of buying an edge in distance $2$, and for Lemma~3.3 the operation of deleting an edge is considered.
Both operations are allowed with $k\geq 2$.
Hence, for any $k\geq 2$, we have $\kNE = \NE$ and the PoA bound of $1.5$ from \cite{Al06} applies.
\end{proof}

\end{document}